\documentclass[manyauthors]{fundam}

\usepackage{latexsym}
\usepackage{amsfonts}
\usepackage{amssymb}
\usepackage{amsmath}
\usepackage{mathtools}
\usepackage[T1]{fontenc}
\usepackage[utf8]{inputenc}
\usepackage{tikz}
\usetikzlibrary{arrows}
\usetikzlibrary{backgrounds}
\usetikzlibrary{patterns}
\usetikzlibrary{fit}
\usetikzlibrary{shapes.geometric}
\usetikzlibrary{shapes.symbols}
\usetikzlibrary{decorations.pathmorphing}
\usepackage{lineno} 
\usetikzlibrary{decorations}
\usepackage{enumitem} \setlist[description]
{labelindent=3mm,labelwidth=10mm,leftmargin=15mm,itemsep=1mm}
\setlist[itemize]
{labelindent=3mm,labelwidth=10mm,leftmargin=15mm,itemsep=1mm}
\pgfarrowsdeclare{arrow30}{arrow30}
{
 \arrowsize=0.3pt
 \advance\arrowsize by.275\pgflinewidth%
 \pgfarrowsleftextend{+-\arrowsize}
 \advance\arrowsize by.5\pgflinewidth
 \pgfarrowsrightextend{+\arrowsize}
}
{
 \arrowsize=0.3pt
 \advance\arrowsize by.275\pgflinewidth%
 \pgfsetdash{}{+0pt}
 \pgfsetroundjoin
 \pgfpathmoveto{\pgfqpoint{1\arrowsize}{0\arrowsize}}
 \pgfpathlineto{\pgfqpoint{-12\arrowsize}{2.5\arrowsize}}
 \pgfpathlineto{\pgfqpoint{-12\arrowsize}{-2.5\arrowsize}}
 \pgfpathclose
 \pgfusepathqfillstroke
}

\DeclareMathOperator{\noidx}  {noidx}
\DeclareMathOperator{\income} {in}
\DeclareMathOperator{\sign}   {sign}
\DeclareMathOperator{\paths}  {paths}
\DeclareMathOperator{\mult}   {mult}
\DeclareMathOperator{\supp}   {supp}
\DeclareMathOperator{\PRE}    {pre}
\DeclareMathOperator{\POST}   {post}
\DeclareMathOperator{\reach}  {reach}
\DeclareMathOperator{\eff}    {eff}
\DeclareMathOperator{\pred}   {pred}

\newcommand{\STEP}[2]    {\boldsymbol{[}#1\boldsymbol{\rangle}_{\!#2}\,}

\newcommand{\reverse}[1] {\overline{#1}}
\newcommand{\IDX}[1]     {\langle #1\rangle}
\newcommand{\munion}     {\sqcup}
\newcommand{\Idx}        {\textit{idx}}
\newcommand{\rev}        {\mathit{rev}}
\newcommand{\seq}        {\mathit{seq}}
\newcommand{\set}        {\mathit{set}}
\newcommand{\spike}      {\mathit{spike}}

\newcommand{\setrev}     {\mathit{srev}}
\newcommand{\mixrev}     {\mathit{mrev}}
\newcommand{\splitrev}   {\mathit{split}}
\newcommand{\STS}        {\mathit{STS}}
\newcommand{\CRG}        {\mathit{CRG}}
\newcommand{\es}         {\varnothing}
\newcommand{\ralpha}     {\reverse{\alpha}}
\newcommand{\rbeta}      {\reverse{\beta}}

\newcommand{\rT}         {\reverse{T}}

\newcommand{\rSTS}       {\reverse{\STS}}
\newcommand{\wN}         {{\widetilde N}}
\newcommand{\wF}         {\widetilde F}
\newcommand{\wM}         {\widetilde M}
\newcommand{\wP}         {\widetilde P}
\newcommand{\xright}[1]  {\xrightarrow{~#1~}}
\newcommand{\bt}         {\blacktriangleleft}
\newcommand{\proj}[2]    {#1|_{#2}}

\usepackage{hyperref}

\begin{document}

\setcounter{page}{67}
\publyear{2021}
\papernumber{2082}
\volume{183}
\issue{1-2}

  \finalVersionForARXIV

\title{Investigating Reversibility of Steps in Petri Nets}

\author{David~de~Frutos~Escrig
         \\
			 Facultad de Ciencias Matem{\'a}ticas\\
			 Universidad Complutense de Madrid\\
             28040 Madrid, Spain\\
             defrutos@sip.ucm.es
\and
             Maciej Koutny\\
             School of Computing\\
             Newcastle University \\
             Newcastle upon Tyne NE1 7RU, U.K.\\
             maciej.koutny@ncl.ac.uk \vspace*{2mm}\\
\and
             { \L }ukasz Mikulski\thanks{Address for correspondence:  Faculty of Mathematics and Computer Science,
             Nicolaus Copernicus University,  Chopina 12/18, 87-100 Toru\'n, Poland. \newline Also affiliated at: Institute of Computer
             Science, Polish Academy of Sciences, Jana Kazimierza 5, 01-248 Warsaw, Poland}
                    \\
             Faculty of Mathematics and Computer Science\\
             Nicolaus Copernicus University\\
             Chopina 12/18, 87-100 Toru\'n, Poland\\
             lukasz.mikulski@mat.umk.pl
			}

\runninghead{D.~de~Frutos~Escrig et al.}{Investigating Reversibility of Steps in Petri Nets}

\maketitle

\vspace*{-6mm}
\begin{abstract}
In reversible computations one is interested in the development of
mechanisms allowing
to undo the effects of executed actions.
The past research has been concerned mainly with reversing single actions.
In this paper, we consider the problem of reversing the effect of
the execution of groups of actions (steps).

Using Petri nets as a system model,
we introduce concepts related to this new scenario, generalising
notions used in the single action case. We then present
properties arising when reverse
actions are allowed in place/transition nets (\textsc{pt}-nets).
We obtain both positive
and negative results, showing that allowing steps makes
reversibility more problematic than in the interleaving/sequential case.
In particular, we demonstrate that there is a crucial difference
between reversing steps which are sets and those which are
true multisets.
Moreover, in contrast to sequential semantics, splitting reverses does not
lead to a general method for reversing bounded \textsc{pt}-nets.
We then show that a suitable solution can be obtained by
combining split reverses with weighted read arcs.
\end{abstract}

\section{Introduction}

Reversibility of (partial) computations
has been extensively studied during the past years,
looking for mechanisms that allow to (partially) undo some actions executed during
a computational process, that for some reason one needs to cancel.
As a result, the execution can then continue
from a consistent state as if that suppressed action had not been executed at all.
In particular, these mechanisms allow for the correct implementation of
transactions~\cite{DanKri04,DanKri05},
that are partial computations which either are
totally executed or not executed at
all. This includes updating
in databases, so that one never commits
an `incomplete' set of related updates
that might produce an inconsistent state
(in which one could infer contradictory
facts).
Another example would be money transfers
between banks,
or modern
e-commerce platforms, where the payments received
should match the goods distributed~\cite{cohen2013systems}.

\medskip
Within  Formal Methods, reversibility has been
investigated, for instance, in the
framework of process calculi~\cite{PhiUli07,LanMezSte10},
event structures~\cite{PhiUl15}, DNA-computing~\cite{CarLan11},
category theory~\cite{DanKriSob07},
and quantum
computing~\cite{DBLP:books/daglib/0025734}.
In the latter case, it plays a central role due to
the inherent reversibility of the mechanisms on
which quantum computing is based.
This paper is concerned with reversibility in
place/transition nets (\emph{\textsc{pt}-nets}),
which are a fundamental class of \emph{Petri nets},
operating according to the step semantics in
which multisets of actions (\emph{steps})
are executed simultaneously.

In Petri nets, reversibility is usually understood as a global property
resembling cyclicity.
It was also considered in a manner
closer to its process calculi meaning using symmetric nets~\cite{EspNie94}
(symmetric nets have later been used
to study structural symmetries  of state spaces~\cite{Kor11}).
Locally defined reversibility has not yet been extensively studied
within the Petri net framework.
This is rather surprising as the formalisation
of an action by means of a pair of
\emph{pre-places} and \emph{post-places}
provides an immediate way of defining the
\emph{reverse} of the actions  simply by
interchanging these two sets of places.
There are, however, some more recent
works in which reversibility is understood as cyclicity
(i.e., an ability to return to the initial state from any reachable state).
They are usually based on the structure theory of Petri nets~\cite{HujDelKor15},
or an algebraic study by means of invariants~\cite{ApaAyb08}.

From the operational point of view, one can distinguish three essential ways
of reversing computational processes: backtracking, causal reversibility, and
out of causal reversibility.
For concurrent systems, the backtracking mode was
considered, for example, in~\cite{DanKri04}, where the RCCS process algebra
is introduced.
An investigation of causal reversibility in the Petri net context  can
be found, for example, in~\cite{DBLP:conf/coordination/MelgrattiMU19}, where
it was implemented using occurrence nets.
All three ways of reversing computations were studied
in~\cite{DBLP:conf/rc/PhilippouP18}, where biologically motivated
reversing Petri nets were introduced.
In all these works, one needs to enrich the original model by
additional annotations or constructs.
It is the memory of monitored processes for RCCS, the computation stack encoded
through colours for folded occurrence nets, and atoms and bonds together with the
history function for reversing Petri nets.
In our approach, we are interested in studying the possibility of
reversing computations in step semantics emphasizing reversing the effects,
and avoiding the reachability of new states.
The latter ensures that one can reach only states that are reachable
by forward computations,
which differentiates  our approach from the out of causal reversibility discussed
in~\cite{DBLP:conf/rc/PhilippouP18}.
We also do not equip our nets with
additional external monitors which help to ensure causality.
As a result,
it may happen that reverses of
actions that were not yet executed become enabled.
This inconvenience can, however, be easily removed by suitably
augmenting a \textsc{pt}-net
being reversed to yield another net, as described in~\cite{Bologna}.

The approach presented in this paper is closer to inverse nets presented
in~\cite{BouFin97}, and so more \emph{operational}.
It extends the study of reversing (sequential)
transition systems initiated in~\cite{Bologna}, where
it was shown that the apparent simplicity of this
approach is far from trivial, mainly due to the difficulty
of avoiding situations where an added reverse action is
executed in an inconsistent manner.
Further investigation of
this problem can be found in~\cite{LanMik},
while~\cite{csp2016}
considers  \emph{bounded}
\textsc{pt}-nets, distinguishing between the
\emph{strict} reverses and \emph{effect} reverses of actions.
The latter deliver the effect of reversing the original actions,
but possibly with a change in the way action enabling is
carried out.
It was shown that some transition systems
which can be \emph{solved} by bounded nets allow
the reversal of their actions by means of single
reverse actions, while in other cases
the reversal is only possible if
\emph{splitting} of reverses is allowed
(i.e., each action has a set of reverses which collectively
provide means of reversing the original action).

In~\cite{csp2016} only the sequential (\emph{interleaving})
semantics of nets was considered and, in
fact, several of the presented examples were just (finite)
\emph{linear transition systems}, taking advantage of the
results presented in~\cite{ATAED,PSC}, where binary words representable
by Petri net were characterised.
The latter problem and its consequences
for reversibility has been further investigated in~\cite{PN18}.\smallskip

\medskip\noindent
\textbf{About this paper~~}
We consolidate and extend the results of~\cite{pn2019}, where the study of
\emph{step reversing} in \textsc{pt}-nets and (step) transition systems
was initiated.
We assume that the transition systems to be
synthesized include information about
the multisets of  actions (steps) that should be
executed  in parallel. Reversing of the
actions should preserve this step information so
that the simultaneous firing of several reverse
actions should correspond
to the original steps at the system represented by a \textsc{pt}-net.

We introduce several concepts related to this new scenario,
generalising notions used in the single action case.
A number of straightforward definition which worked in the sequential case
are no longer adequate.
When looking for their adequate generalisations,
we  identify two `natural' notions
of step reversibility.
The former (\emph{direct reversibility}) only allows
steps which comprise either
the original actions, or the reverse actions.
The latter (\emph{mixed reversibility})
allows also mixing of the original and reverse actions.
It turns out that these two ways of interpreting step
reversibility are fundamentally different.
Crucially, the direct reversibility
cannot be implemented for steps which are true multisets,
and so in such cases one has to look for mixed reversibility solutions.
In this way, we identified a striking difference
between reversing steps which are sets and those which are
true multisets (when autoconcurrency
of actions in system executions is allowed).
However, there is still a
general positive result which basically applies whenever
sequential reversing is possible and
the original steps can be
be satisfactorily represented.

We also adapt split reverses introduced in~\cite{csp2016}.
Unfortunately, splitting is not enough to deal with all bounded \textsc{pt}-nets
(also adding inhibitor arcs to the \textsc{pt}-net model does not always help).
A general solution we propose uses \emph{weighted read arcs}~\cite{KleKou02}
(the further development of this model is out of the scope of this
paper, and is left as a topic for the future work).

The paper is organised as follows.
Section~\ref{sect-0} recalls
notions and notations used throughout the paper.
Moreover, some basic results concerning the step transition
model are given.
Section~\ref{se-ddd} introduces four different ways of
defining reversibility in step transition systems,
including direct step reversibility and mixed step reversibility,
as well as set reversibility (where a true multiset of actions is reversed
in stages) and split reversibility.
Section~\ref{sect-1} demonstrates that the direct reversibility
cannot be achieved in the presence of autoconcurrency.
Moreover, it characterises cases where mixed reversibility
can be replaced by (more desirable)
direct reversibility or set reversibility.
Section~\ref{sect-2} provides result allowing one
to deal with mixed reversibility and step
reversibility in an effective way, by reducing the
reversibility problem to the net synthesis
problem. This approach is further continued Section~\ref{sect-3}, where
lifting of sequential reversibility to
step reversibility is discussed.
Section~\ref{sect-split}  proposes a general solution to the step
reversibility of bounded \textsc{pt}-nets which relies on
the weighted read arcs.
Finally, Section~\ref{sect-tdt}
contains concluding remarks.

\section{Preliminaries}
\label{sect-0}

\paragraph{Vectors, multisets and actions}

An \emph{$X$-vector} over a set $X$ is a mapping
$\alpha:X\to\mathbb{Z}$, where
$\mathbb{Z}$ is the set of all integers.
For two $X$-vectors, $\alpha$ and $\beta$,
the
\emph{sum} ($\alpha+\beta$),
\emph{difference} ($\alpha-\beta$),
and
\emph{less-than-or-equal} relationship ($\alpha\leq\beta$)
are defined component-wise.
The \emph{support} of an $X$-vector $\alpha$ is the set
$\supp(\alpha)=\{x\in X\mid \alpha(x)\neq 0\}$.
The \emph{empty} $X$-vector
has the empty support and is denoted by $\es_X$ or simply by $\es$, and
$-\alpha$ denotes $\es_X-\alpha$.
The \emph{union} of an $X$-vector $\alpha$ and
a $Y$-vector $\beta$, where $X\cap Y=\es$,
is the $(X\cup Y)$-vector
$\alpha\munion\beta$ such that
$\alpha\munion\beta|_X=\alpha$ and
$\alpha\munion\beta|_Y=\beta$.

\medskip
\emph{Multisets} over $X$ are $X$-vectors returning
non-negative integers in $\mathbb{N}$, the subsets of $X$ can be identified
with multisets returning $0$ or $1$, and
the elements of $X$ with singleton sets.
The set of all multisets over $X$ is denoted by $\mult(X)$.
The \emph{size} of
$\alpha\in\mult(X)$ is given by
$|\alpha|=\sum_{x\in X}\alpha(x)$.
For $x\in X$, we denote $x\in\alpha$ whenever $\alpha(x)\geq 1$.

In what follows, e.g.,
$(xxz)$ denotes a multiset $\alpha$ with the support $\{x,z\}$
satisfying $\alpha(x)=2$
and  $\alpha(z)=1$.
Moreover,
$x^k$ denotes a multiset $\alpha$ with the support $\{x\}$ satisfying
$\alpha(x)=k$.

Throughout the paper, $\mathcal A$ denotes an infinite set  \emph{actions},
including the \emph{reverse actions} and \emph{indexed reverse actions} introduced
in Section~\ref{se-ddd},
used in step transition systems and \textsc{pt}-nets
to model events occurring in concurrent behaviours.
To simplify the presentation, we will treat a vector or
multiset $\alpha$ over $T\subseteq\mathcal A$ as a vector or
multiset over $\mathcal A$, assuming that
$\alpha|_{\mathcal A\setminus T}=\es_{\mathcal A\setminus T}$.

\paragraph{Step transition systems}

A \emph{step transition system} is a tuple
$\STS=(S,T,\to,s_0)$ such that
$S$ is a nonempty set of \emph{states}, $T$ is a finite
set of \emph{actions}, $\to\, \subseteq S \times \mult(T) \times S$
is the set of \emph{transitions},
and $s_0\in S$ is the \emph{initial state}.
The transition labels in $\mult(T)$ represent simultaneous
executions of groups of actions, called \emph{steps}.
Rather than
$(s,\alpha,r)\in\;\to$, we can denote
$s\xright{\alpha}_\STS r$.
Moreover, $s\xright{\alpha}_\STS$
means that there is some $r$
such that $s\xright{\alpha}_\STS r$.
$\STS$ is:
\begin{itemize}
\item
    a \emph{set transition system} if $\alpha$ is a set, for
    every transition $(s,\alpha,r)$; and
\item
    \emph{state-finite} if $S$ is finite,
    \emph{step-finite} if $\{\alpha\mid s\xright{\alpha}_\STS\}$
    is finite, and \emph{finite} if it is both state- and step-finite
   (and so $\to$ is finite).
\end{itemize}
In the diagrams, step transition systems are depicted as labelled
directed graphs. Arcs labelled by the empty multiset are
omitted.

A state $r$ is \emph{reachable} from state $s$ if there are
steps $\alpha_1,\dots,\alpha_k$ ($k\geq 0$) and
states $s_1,\dots, s_{k+1}$ such that
$(s=)s_1\xright{\alpha_1}_\STS s_2\dots s_k\xright{\alpha_k}_\STS s_{k+1}(=r)$.
We denote this by
$s\xright{\alpha_1\cdots\alpha_k}_\STS r$.

The set of all states from which a state $s$ is reachable
is denoted by $\pred_\STS(s)$,
$s$ is a \emph{home state} if $\pred_\STS(s)=S$, and
$R\subseteq S$ is a
\emph{home cover} of $\STS$ if $S=\bigcup_{s\in R}\pred_\STS(s)$.

An \emph{(undirected) path} from
a \emph{source} state $s$ to \emph{target} state $r$
is a sequence $\pi=\tau_1\dots\tau_k$ ($k\geq 0$), where
each $\tau_i$ is a pair
$((s_i,\alpha_i,r_i),\zeta_i)\in(\to\times\{+,-\})$
such that  either $k=0$ and $s=r$,
or $k\geq 1$
and $s=\widehat s_1,\widehat r_1=\widehat s_2,
\dots,\widehat r_{k-1}=\widehat s_k,\widehat r_k=r$,
assuming that
$\widehat s_i=s_i$ and $\widehat r_i=r_i$
if $\zeta_i=+$, and otherwise $\widehat s_i=r_i$ and $\widehat r_i=s_i$,
for every $1\leq i\leq k$.
We denote this by $\pi\in\paths_\STS(s,r)$.
The \emph{signature} of $\pi$ is the
$\mathcal A$-vector  $\sign(\pi)=\es_{\mathcal A} \zeta_1 \alpha_1 \dots \zeta_k\alpha_k$,
where the $\zeta_i$'s are being treated as addition and subtraction operations.
For example,
if $\pi=((s',\alpha,s),-)((s',\beta,s''),+)\in\paths_\STS(s,s'')$, then
$\sign(\pi)=\es_{\mathcal A}-\alpha+\beta=\beta-\alpha$.

Intuitively,
$\sign(\pi)$ records the `net contribution (or effect)'
made by each action along the path $\pi$, with $a\in\alpha_i$
making a `positive' contribution
if the transition $(s_i,\alpha_i,r_i)$ agrees with the direction of the path,
and otherwise making a `negative' contribution.
Note that $r$ is reachable from $s$ iff there is
$\pi\in\paths_\STS(s,r)$ with all the $\zeta_i$'s being equal to $+$.

In this paper, step transition systems are
intended to capture (step) reachability graphs of \textsc{pt}-nets.
We will now introduce a property of step transition systems which is
motivated by the \emph{state equation} which holds,
in particular, for \textsc{pt}-nets.
The basic idea is that the effect of executing
an action is fixed, and so
does not depend on the global state in which this happens
(we will make this more precise later).
Capturing such a constant effect is straightforward for
\textsc{pt}-nets, but not for step transition systems.
One can, however,
approximate the concept of having `the same effect' by considering
as equivalent all
undirected paths with the same source and target states.

\medskip
Let $\bowtie_\STS$ be the least equivalence
relation on the set of all $\mathcal A$-vectors such that:
\textit{(i)}
$\sign(\pi)\bowtie_\STS\sign(\pi')$, for
all $s,r\in S$ and $\pi,\pi'\in\paths_\STS(s,r)$;
and
\textit{(ii)}
$\alpha\bowtie_\STS\beta$ and $\alpha'\bowtie_\STS\beta'$
imply
$\alpha+\alpha'\bowtie_\STS\beta+\beta'$, for all
$\mathcal A$-vectors $\alpha$, $\alpha'$, $\beta$, and $\beta'$.
Intuitively, $\alpha\bowtie_\STS\beta$ means that
executing $\alpha$
has the same effect as
executing $\beta$.
This leads to the following property of a step transition
$\STS$:
\begin{description}
\item[\textit{CE}]
    $\sign(\pi)\bowtie_\STS\sign(\pi')$
    implies $r=r'$, for all $s,r,r'\in S$,
    $\pi\in\paths_\STS(s,r)$, and $\pi'\in\paths_\STS(s,r')$.
    \hfill
    \emph{(constant effect)}
\end{description}
It is the case that $\alpha\bowtie_\STS\beta$ implies
$-\alpha\bowtie_\STS-\beta$ since
$\pi\in\paths_\STS(s,r)$ means that there is
$\pi'\in\paths_\STS(r,s)$ such that $\sign(\pi')=-\sign(\pi)$.
Hence we also have the following
`backward' version of the `forward' constant effect property \textit{CE}:
$\sign(\pi)\bowtie_\STS\sign(\pi')$
implies $s=s'$, for all $s,s',r\in S$,
$\pi\in\paths_\STS(s,r)$, and $\pi'\in\paths_\STS(s',r)$.

\medskip
We are now in a position to introduce a class of step transition systems
used throughout the rest of this paper. A step transition system $\STS=(S,T,\to,s_0)$ is
a \emph{constant effect step transition system} (or \textsc{cest}-system)
if it satisfies \textit{CE} as well as the following three properties,
for every $s\in S$:
\begin{description}
\itemsep=0.9pt
\item[\textit{REA}]
    $s_0\in \pred_\STS(s)$.
    \hfill
    \emph{(reachability)}

\item[\textit{EL}]
    $s\xright{\es}_\STS s$.
    \hfill
    \emph{(empty loops)}

\item[\textit{SEQ}]
    $s\xright{\alpha+\beta}_\STS$ implies
    $s\xright{\alpha\beta}_\STS$.
    \hfill
    \emph{(sequentialisability)}
\end{description}

\medskip
We then obtain two immediate properties of \textsc{cest}-systems.

\begin{proposition}
\label{prop-zt}
    Let $\STS$ be a \textsc{cest}-system.
\begin{enumerate}
\itemsep=0.9pt
\item
    $r=r'$ whenever
    $s\xright{\alpha}_\STS r$ and
    $s\xright{\alpha}_\STS r'$.
\eject
\item
    $s=r$ whenever $s\xright{\es}_\STS r$.
\end{enumerate}
\end{proposition}
\begin{proof}
    Part (1)
    follows from \textit{CE}, and
    part (2)
    follows from part (1) and
    \textit{EL}.
\end{proof}
Proposition~\ref{prop-zt}(1)
captures the property of \emph{forward determinism} (\textit{FD})
which
allows one to unambiguously  denote $s\oplus_\STS\alpha$, or
$s\oplus\alpha$ if $\STS$ is clear from the context,
as the state $r$ satisfying $s\xright{\alpha}_\STS r$
whenever $s\xright{\alpha}_\STS$.

Being a \textsc{cest}-system
still does not mean that it can be
generated by a
\textsc{pt}-net.
A complete characterisation
can be obtained using, e.g.,
theory of
regions~\cite{DBLP:conf/ac/BadouelD96,DBLP:journals/fuin/DarondeauKPY09}.

\begin{proposition}
\label{prop-ffffp}
    Let $s$ be a state of a
    \textsc{cest}-system $\STS$.
    If
    $s\oplus\alpha$ is defined
    and
    $\beta+\gamma\leq\alpha$,
    then
    $s\oplus\beta,s\oplus(\beta+\gamma)$ and $(s\oplus\beta)\oplus\gamma$
    are also defined,
    and $(s\oplus\beta)\oplus\gamma=s\oplus(\beta+\gamma)$.
\end{proposition}
\begin{proof}
    By $s\xright{\alpha}_\STS $ as well as  \textit{SEQ} and \textit{CE},
    we have
    $s\xright{\beta}_\STS s\oplus\beta\xright{\gamma}_\STS (s\oplus\beta)\oplus\gamma$
    as well as
    $s\xright{\beta+\gamma}_\STS s\oplus(\beta+\gamma)$.
    We therefore have
    $\pi=((s,\beta,s\oplus\beta),+)((s\oplus\beta,\gamma,(s\oplus\beta)\oplus\gamma),+)
     \in\paths_\STS(s,(s\oplus\beta)\oplus\gamma)$
    and
    $\pi'=((s,\beta+\gamma,s\oplus(\beta+\gamma)),+)
     \in\paths_\STS(s,s\oplus(\beta+\gamma))$.
    Moreover, $\sign(\pi)=\beta+\gamma=\sign(\pi')$.
    Hence, by \textit{CE},
    $(s\oplus\beta)\oplus\gamma=s\oplus(\beta+\gamma)$.
\end{proof}

We use different ways of removing transitions from
a step transition system $\STS=(S,T,\to,s_0)$:
\[
\begin{array}{lcl@{\,}ll}
    \STS^\seq
    & =
    & (S,T,
    & \{(s,\alpha,r)\in \;\to
      ~\mid~  |\alpha|\leq 1\}, s_0)
\\
    \STS^\set
    & =
    & (S,T,
    & \{(s,\alpha,r)\in \;\to
      ~\mid~  \supp(\alpha)=\alpha\},  s_0)
\\
    \STS^\spike
    & =
    & (S,T,
    & \{(s,\alpha,r)\in \;\to
     ~\mid~  |\supp(\alpha)|\leq 1\}, s_0)
\\
    \STS|_{T'}
    & =
    & (S,T',
    & \{(s,\alpha,r)\in \;\to
      ~\mid~ \alpha\in\mult(T')\}, s_0)    ~~~~~~ (\textit{for}~~T'\subseteq T)\;.
\end{array}
\]
That is, $\STS^\seq$ is obtained by only
retaining singleton steps and $\es$-labelled steps,
$\STS^\set$ by only retaining steps which
are sets, and
$\STS^\spike$ by removing all steps which
use more than one action.
Moreover, $\STS$ is a \emph{sequential} / \emph{set}
/ \emph{spiking} step transition system if
respectively $\STS=\STS^\seq$ /
$\STS=\STS^\set$ /
$\STS=\STS^\spike$.\footnote
{
    If $\STS$ is a \textsc{cest}-system,
    then $\STS^\seq$, $\STS^\set$, and $\STS^\spike$
    satisfy \textit{REA} since $\STS$ satisfies \textit{REA} and \textit{SEQ}.
}

\medskip
For step transition systems satisfying
\textit{SEQ}, checking the satisfaction of the constant
effect property can be done by restricting oneself to the
sequential steps.

\begin{proposition}
\label{prop-x1}
    Let $\STS$ be a step transition system satisfying \textit{SEQ}.
    Then
    $\STS$ satisfies \textit{CE} if and only if
    $\STS^\seq$ satisfies \textit{CE}.
\end{proposition}
\begin{proof}
    We first observe that from
    \textit{SEQ} for $\STS$ it follows that,
    for every $\pi\in\paths_\STS(s,r)$, there is
    $\pi'\in\paths_{\STS^\seq}(s,r)$ such that $\sign(\pi')=\sign(\pi)$
    \textit{(*)}.
    Hence, we also have  $\bowtie_\STS\,=\,\bowtie_{\STS^\seq}$
    \textit{(**)}.

\medskip
($\Longrightarrow$)
    Follows from \textit{(**)} and
    $\pi\in\paths_{\STS^\seq}(s,r)\subseteq\pi\in\paths_\STS(s,r)$.

($\Longleftarrow$)
    Follows from \textit{(*)} and \textit{(**)}.
\end{proof}

The essence of the next result is that adding reverses of some
transitions labelled by the same action
in a sequential step transition system preserves
the constant effect property.

\begin{proposition}
\label{prop-x2}
    Let $\STS=(S,T,\to,s_0)$ be a sequential
    step transition system satisfying \textit{CE} and
    $\STS'=(S,T\cup\{\widetilde{a}\},\to\cup\to',s_0)$,
    where $\to'\subseteq\{(r,\widetilde{a},s)\mid (s,a,r)\in\to\}$
    for some $a\in T$ and $\widetilde{a}\notin T$.
    Then $\STS'$ satisfies \textit{CE}.
\end{proposition}
\begin{proof}
    The result clearly holds when $\to'$ is empty.
    Otherwise, we have $a\bowtie_{\STS'}-\widetilde{a}$.
    For every $\mathcal A$-vector $\alpha$, let
    $\widehat{\alpha}$ be the $\mathcal A$-vector
    such that
    $\widehat{\alpha}|_{\mathcal A\setminus\{a,\widetilde{a}\}}
     =
     \alpha|_{\mathcal A\setminus\{a,\widetilde{a}\}}$,
    $\widehat{\alpha}(a)=\alpha(a)-\alpha(\widetilde{a})$, and
    $\widehat{\alpha}(\widetilde{a})=0$.

    We observe that, for all $s,r\in S$ and $\pi\in\paths_{\STS'}(s,r)$, there is
    $\pi'\in\paths_\STS(s,r)$
    such that $\sign(\pi')=\widehat{\sign(\pi)}$ \textit{(*)}.
    Hence, we also have that $\alpha\bowtie_{\STS'}\beta$
    iff $\widehat{\alpha}\bowtie_\STS\widehat{\beta}$,
    for all $\mathcal A$-vectors $\alpha$ and $\beta$
    \textit{(**)}.
    The result then follows from \textit{CE} for $\STS$ together with
    \textit{(*)} and \textit{(**)}.
\end{proof}

Let $\STS=(S,T,\to,s_0)$ and $\STS'=(S',T',\to',s'_0)$ be
two step transition systems such that $T\subseteq T'$.
Then $\STS$ is \emph{included} in $\STS'$
if there is a bijection $\psi \colon S \to S'$
such that $\psi(s_0)=s'_0$ and
$\{(\psi(s),\alpha,\psi(s'))\mid s\xright{\alpha}_\STS s'\}
 \;\subseteq\;\to'$.\footnote
{
    If $\STS$ and $\STS'$ are \textsc{cest}-systems, then
    $\psi$ is unique due to
    \textit{REA} and \textit{FD}.
}
This is denoted by
$\STS\lhd_\psi \STS'$ or
$\STS\lhd \STS'$, and if $\psi$ is the identity
on $S$, we denote $\STS\bt \STS'$.
Also,  $\STS$ is \emph{isomorphic} with $\STS'$
if there is $\psi$ such that $\STS\lhd_\psi \STS'$
and
$\STS'\lhd_{\psi^{-1}} \STS$.
This is denoted by
$\STS\simeq_\psi \STS'$ or $\STS\simeq \STS'$.

\paragraph{\textsc{pt}-nets}

A \emph{\textsc{pt}-net} (short for place/transition net
\cite{DBLP:books/daglib/0032298}) is
a tuple $N=(P,T,F,M_0)$, where $P$ is a finite set of \emph{places},
$T\subseteq \mathcal A$ is a disjoint finite set of \emph{actions},\footnote
{
    We use the term `actions' rather than `transitions' when
    referring to the elements of $T$,
    in order to avoid confusion with the triples $(s,\alpha,r)$
    used in the definition of step transition systems.
}
$F$ is the \emph{flow function} $F\colon(P\times T)\cup(T\times P)\to\mathbb{N}$
specifying the arc weights between places and actions,
and $M_0$ is the \emph{initial marking}
(\emph{markings} are multisets over $P$ representing global states).
It is assumed that, for every $a\in T$, there is $p\in P$
such that $F(p,a)>0$.

The triple $(P,T,F)$ is an \emph{unmarked} \textsc{pt}-net, and
$N|_{T'}=(P,T',F|_{(P\times T')\cup(T'\times P)}, M_0)$
is the  \emph{subnet} of $N$ induced by $T'\subseteq T$.

In the diagrams,  \textsc{pt}-nets are depicted as labelled
directed graphs, with circles  representing places and
boxes to representing actions. Markings are represented by
black tokens or numbers drawn inside the circles,
the arc weight of 1 is omitted, and the 0-weight arcs are not drawn.

\medskip
Multisets over $T$,
again called  \emph{steps},
represent executions of groups of actions.
The \emph{effect} of a step $\alpha\in\mult(T)$
(and, in general, a $T$-vector $\alpha$)
is the $P$-vector
$\eff_N(\alpha)=\POST_N(\alpha)-\PRE_N(\alpha)$,
where
$\PRE_N(\alpha)$ and $\POST_N(\alpha)$
are multisets of places
such that, for every $p\in P$:
\[
    \PRE_N(\alpha)(p) =\sum_{a\in T}\alpha(a)\cdot F(p,a)
    ~~~~~~\textrm{and}~~~~~~
    \POST_N(\alpha)(p)=\sum_{a\in T}\alpha(a)\cdot F(a,p)\;.
\]

A step $\alpha$ is \emph{enabled} at a~marking $M$ if $\PRE_N(\alpha)\leq M$, and
the \emph{firing} of such a step leads to the marking
$M'=M+\eff_N(\alpha)$.\footnote
{
    $M'$ is a multiset due to $\PRE_N(\alpha)\leq M$.
}
This is respectively denoted by $M\STEP{\alpha}{N}$ and
$M\STEP{\alpha}{N} M'$.
Note that it is always the case that $M\STEP{\es}{N}M$, and that
$M\STEP{\alpha+\beta}{N}$ implies
$M\STEP{\alpha}{N} M'\STEP{\beta}{N}$, where $M'=M+\eff_N(\alpha)$.
These two facts motivated the inclusion of
\textit{EL} and \textit{SEQ} in the definition of \textsc{cest}-systems.

\medskip
The \emph{reachable} markings of $N$ are
the smallest set of markings $\reach_N$
such that $M_0\in\reach_N$ and if $M\in\reach_N$ and
$M\STEP{\alpha}{N}$,
then $M+\eff_N(\alpha)\in\reach_N$.
$N$ is \emph{bounded} if the set $\reach_N$ of
all the reachable markings is finite.

The overall behaviour of $N$
can be captured by its
\emph{concurrent reachability graph} which is the
step transition system
$\CRG_N =(\reach_N,T,\{(M,\alpha,M')\mid M\in\reach_N\wedge
M\STEP{\alpha}{N} M'\}, M_0)$.
In what follows,
$M\xright{\alpha}_N M'$  denotes
$M\xright{\alpha}_{\CRG_N } M'$.
Note that the
concurrent reachability graphs of bounded \textsc{pt}-nets are finite.

\medskip
The concept of \emph{marking equation} can be explained in the following way.
Suppose that a marking $M'$ can be reached from marking $M$
by firing a sequence of steps, e.g., $M\xright{\alpha_1\cdots\alpha_n}_{\CRG_N} M'$.
Then
\begin{equation}
\label{eq-marking}
    M'=M+\eff_N(\alpha)
    ~~~~ ~~~~~~~~
    M=M'-\eff_N(\alpha)
    ~~~~ ~~~~~~~~
    \eff_N(\alpha)=M'-M\;,
\end{equation}
where $\alpha=\alpha_1+\dots+\alpha_n$.
This means that the \emph{effect} of
executing a multiset of actions $\alpha$
is constant, as it does not
depend on the starting marking nor the ending marking nor
any particular way
in which the actions making up $\alpha$ were fired.
Moreover, the effect of actions fired along any path
from $M$ to $M'$ is constant.
This motivated the inclusion of \textit{CE}
in the definition of \textsc{cest}-systems.

\medskip
It is straightforward to see that $\CRG_N$ is a \textsc{cest}-system.
In particular, by Eq.\eqref{eq-marking}, we have
$\eff_N(\sign(\pi))=M'-M$,
for every $\pi\in\paths_{\CRG_N}(M,M')$. Hence, in particular,
$\alpha\bowtie_{\CRG_N}\beta$ implies $\eff_N(\alpha)=\eff_N(\beta)$.
As a result, \textit{CE} holds.

\paragraph{Solving step transition systems}

A step transition system $\STS$ is \emph{solvable} if there is
a \textsc{pt}-net $N$ such that $\STS\simeq\CRG_N$.
This is the standard definition used in several
works concerned with the synthesis of Petri nets from
transition systems.
In this paper, we will also use a more general notion of solvability,
defined for step transition systems with multiple initial states.

\medskip
A \emph{step transition system with multiple initial states} is a tuple
$\STS=(S,T,\to,S_0)$ such that the first three components are as in the
definition of a step transition system, and $S_0\subseteq S$ is a nonempty
set of initial states. Moreover, for every $r\in S_0$,
$\STS_r=(S_r,T,\to_r,r)$ is the step transition system such that
$S_r=\{s\in S\mid r\in\pred_\STS(s)\}$ and
$\to_r\,=\,\to\cap\,(S_r\times\mult(T)\times S_r)$.
That is, $\STS_r$ is $\STS$ restricted to those states
which are reachable from $r$.

A step transition system with multiple initial states $\STS$ is
\emph{solvable} if there is
an unmarked \textsc{pt}-net $(P,T,F)$ and a
mapping $\psi:S\to\mult(P)$ such that
$\STS_r\simeq_{\psi|_{S_r}}\CRG_{(P,T,F,\psi(r))}$,
for every
$r\in S_0$.
That is, a solution in this case is an unmarked \textsc{pt}-net
which can be `started' in different initial markings,
each such initial marking solving one of the
step transition systems which make up $\STS$.

\begin{example}
    Let us consider $\STS = (\{q_1,\ldots,q_6\},\{a,b,c\},\to,\{q_1,q_2\})$,
    a step transition system with multiple initial states
    depicted in Figure~\ref{fig:(a|b)cab}($a$)
    (for simplicity, all nonempty steps are singletons).

\medskip
    The step transition system $\STS_{q_2}$,
    depicted on Figure~\ref{fig:(a|b)cab}($b$),
    is obtained from $\STS$ by removing all the states which are not
    reachable from $q_2$. $\STS_{q_1}$ is constructed
    in similar way.
    The \textsc{pt}-net $N=(P,T,F,(p_1p_4))$ solving $\STS_{q_1}$
    is depicted on Figure~\ref{fig:(a|b)cab}($c$).
    As $N=(P,T,F,p_2^4+p_4)$
    is a solution for $\STS_{q_2}$,
    it follows that $\STS$ is solvable.
\hfill$\diamondsuit$
\end{example}

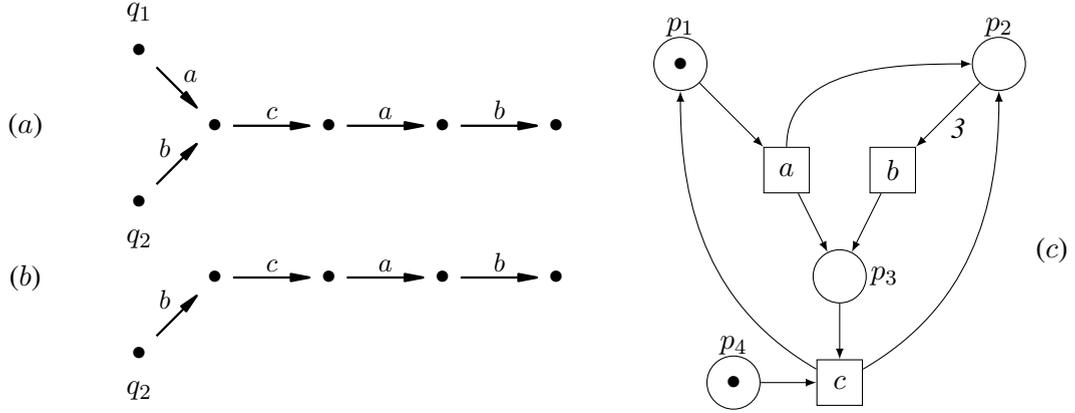
\begin{figure}[h]
\vspace*{-1mm}
\begin{center}
\begin{tikzpicture}[node distance=1.3cm,>=arrow30,%
     line width=0.3mm,scale=1.0,bend angle=45]
     \tikzstyle{box}=[draw,regular polygon,thick,%
     regular polygon sides=4,minimum size=22mm, inner sep = -3pt]
\node (l1) at (-0.5,2.5) {($a$)};

\node (q1) [label=above:$q_1$] at (1,3.5) {$\bullet$};
\node (q2) [label=below:$q_2$] at (1,1.5) {$\bullet$};
\node (q3)   at (2,2.5) {$\bullet$};
\node (q4)   at (3.5,2.5) {$\bullet$};
\node (q5)   at (5,2.5) {$\bullet$};
\node (q6)   at (6.5,2.5) {$\bullet$};

\draw[-arrow30] (q1) to node [auto,inner sep=2pt] {\small $a$} (q3);
\draw[-arrow30] (q2) to node [auto,inner sep=2pt] {\small $b$} (q3);
\draw[-arrow30] (q3) to node [auto,inner sep=2pt] {\small $c$} (q4);
\draw[-arrow30] (q4) to node [auto,inner sep=2pt] {\small $a$} (q5);
\draw[-arrow30] (q5) to node [auto,inner sep=2pt] {\small $b$} (q6);

\node (l2) at (-0.5,0.5) {($b$)};

\node (p2) [label=below:$q_2$] at (1,-0.5) {$\bullet$};
\node (p3)   at (2,0.5) {$\bullet$};
\node (p4)   at (3.5,0.5) {$\bullet$};
\node (p5)   at (5,0.5) {$\bullet$};
\node (p6)   at (6.5,0.5) {$\bullet$};

\draw[-arrow30] (p2) to node [auto,inner sep=2pt] {\small $b$} (p3);
\draw[-arrow30] (p3) to node [auto,inner sep=2pt] {\small $c$} (p4);
\draw[-arrow30] (p4) to node [auto,inner sep=2pt] {\small $a$} (p5);
\draw[-arrow30] (p5) to node [auto,inner sep=2pt] {\small $b$} (p6);
\end{tikzpicture}
~~~~~~~~~
\begin{tikzpicture}[scale=0.7]
\node (l3) at (8,2.5) {($c$)};
\node[circle,draw,minimum size=0.7cm] (pa) at (1,6) {$\bullet$ };
\node (lpa) at (1,6.7) {$p_1$};
\node[circle,draw,minimum size=0.7cm] (pc) at (2,0) {$\bullet$ };
\node (lpc) at (2,.7) {$p_4$};
\node[circle,draw,minimum size=0.7cm] (pb) at (7,6) { };
\node (lpb) at (7,6.7) {$p_2$};
\node[circle,draw,minimum size=0.7cm] (ps) at (4,2) { };
\node (lpb) at (4.7,2) {~~$p_3$};

\node[draw,minimum size=0.6cm] (a) at (3,4){$a$};
\node[draw,minimum size=0.6cm] (b) at (5,4){$b$};
\node[draw,minimum size=0.6cm] (c) at (4,0){$c$};

\draw[-latex] (pa) to node [auto,inner sep=1pt] {} (a);
\draw[-latex] (pb) to node [auto,inner sep=1pt] {\textit{3}} (b);
\draw[-latex] (pc) to node [auto,inner sep=1pt] {} (c);

\draw[-latex] (a) to node [auto,inner sep=1pt] {} (ps);
\draw[-latex] (b) to node [auto,inner sep=1pt] {} (ps);
\draw[-latex] (ps) to node [auto,inner sep=1pt] {} (c);

\draw[-latex,out=150,in=-90] (c) to node [auto,inner sep=1pt] {} (pa);
\draw[-latex,out=30,in=-90] (c) to node [auto,inner sep=1pt] {} (pb);
\draw[-latex,out=90,in=180] (a) to node [auto,inner sep=1pt] {} (pb);
\end{tikzpicture} 
\end{center}\vspace*{-4mm}
\caption{A step transition system with
    multiple initial states $\STS$ ($a$);
    step transition system $\STS_{q_2}$ ($b$);
    and Petri net solving $\STS_{q_1}$ ($c$).
\label{fig:(a|b)cab}}\vspace*{-5mm}
\end{figure}

\section{Reversing steps}
\label{se-ddd}

The reverse action of an action $a$ in a step transition
system $\STS$ or  a \textsc{pt}-net $N$
will be denoted by~$\reverse{a}$.
Intuitively, $\reverse{a}$  cancels the effect of $a$ which
corresponds to   $a+\reverse{a}\bowtie_\STS \es$
and $\eff_N(a)+\eff_N(\reverse{a})=0$, respectively.

\medskip
We consider four ways of
modifying step transition systems to capture the effect of
reversing actions.
In the first three, each action  $a$
has a unique \emph{reverse action} $\reverse{a}$.
Moreover, the reverse $\reverse{\alpha}$
of a multiset $\alpha$ of actions
is obtained by replacing each action
occurrence in $\alpha$ by its reverse.
In the fourth one,
an action $a$ has possibly multiple unique
\emph{indexed reverse actions} $\reverse{a}_{\IDX{\Idx}}$.
The \emph{index-free} version $\noidx(\alpha)$
of a multiset $\alpha$
is obtained by replacing each $\reverse{a}_{\IDX{\Idx}}$
in $\alpha$ by $\reverse{a}$.
For example,
$\noidx((\,\reverse{a}_{\IDX{7}}
  \,\reverse{b}_{\IDX{s,w}}
  \,\reverse{b}\,\reverse{a}_{\IDX{f}}))
  =(\,\reverse{a}\,\reverse{b}\,\reverse{b}\,\reverse{a})
  =\reverse{(abba)}$.

\medskip
In the domain of step transition systems, reversing
is introduced at the behavioural level.
The \emph{direct\,/\,set\,/\,mixed reverse} of
a \textsc{cest}-system
$\STS=(S,T,\to,s_0)$ is respectively
given by:
\[
\begin{array}{l@{\,}c@{\,}l@{~}l@{~}l@{\,}l@{\,}l@{~}ll}
    \STS^\rev
    &=&
    (S,T\uplus\rT,\to\cup\to_\rev, s_0)
    & \textit{with}\;
    & \to_\rev
    &=&
    \{(s\oplus\alpha,\ralpha,s)\mid s\xright{\alpha}_\STS  \}
\\
    \STS^\setrev
    &=&
    (S,T\uplus\rT,\to\cup\to_\setrev, s_0)
    & \textit{with}
    & \to_\setrev
    &=&
    \{(s\oplus\alpha,\ralpha,s)\mid
              s\xright{\alpha}_\STS  \wedge \supp(\alpha)=\alpha\}
\\
    \STS^\mixrev
    &=&
    (S,T\uplus\rT, \to_\mixrev, s_0)
    & \textit{with}
    & \to_\mixrev
    &=&
    \{(s\oplus\alpha,\ralpha+\beta,
             s\oplus\beta)\mid s\xright{\alpha+\beta}_\STS \} \;.
\end{array}
\]
That is, $\to_\rev$ reverses \emph{all} the (original) \emph{steps},
$\to_\setrev$ \emph{only} reverses the steps that are \emph{sets}, and
$\to_\mixrev$ introduces \emph{partial} reverses with \emph{mixed}
steps, including both the original and reverse  actions.
Figure~\ref{fig-0} illustrates mixed reversing.
Note that $s\oplus\alpha$ and $s\oplus\beta$  are states
in $\STS$ due to \textit{SEQ} and \textit{CE}.

\begin{figure}[h]
\begin{center}
\begin{tikzpicture}[node distance=1.3cm,>=arrow30,%
     line width=0.3mm,scale=1.0,bend angle=45]
     \tikzstyle{box}=[draw,regular polygon,thick,%
     regular polygon sides=4,minimum size=22mm, inner sep = -3pt]

\node (s) [label=left:{$s$}]                     at (0,0) {$\bullet$};
\node (sa) [label=$s\oplus\alpha$]                 at (6,-1.5) {$\bullet$};
\node (sab)[label=right:{$s\oplus(\alpha+\beta)$}] at (7,-3) {$\bullet$};
\node (sb) [label=below:$s\oplus\beta$]            at (3,-1.5) {$\bullet$};

\draw[-arrow30,color=black] (s) to [out=-60,in=180] node
                         [auto,swap,inner sep=1pt] {\small $\alpha+\beta$} (sab);
\draw[-arrow30] (sa) to node [auto,swap,inner sep=1pt]
{\small $\ralpha+\beta\,\,\,\,\,\,\,\,\,\,\,\,\,\,$} (sb);
\draw[-arrow30,color=black] (s) to node [auto,inner sep=0pt] {\small $\beta$} (sb);
\draw[-arrow30,color=black] (s) to node [auto,inner sep=1pt] {\small $\alpha$} (sa);
\draw[-arrow30,color=black] (sa) to node [auto,inner sep=1pt] {\small $\beta$} (sab);
\draw[-arrow30,color=black] (sb) to node [auto,inner sep=1pt] {\small $\alpha$} (sab);
\end{tikzpicture}
\end{center}\vspace{-6mm}
\caption{
    A mixed reverse transition
    $s\oplus\alpha\xright{\ralpha+\beta}_\mixrev s\oplus\beta$
    derived from $s\xright{\alpha+\beta}_\STS $.
\label{fig-0}}
\end{figure}
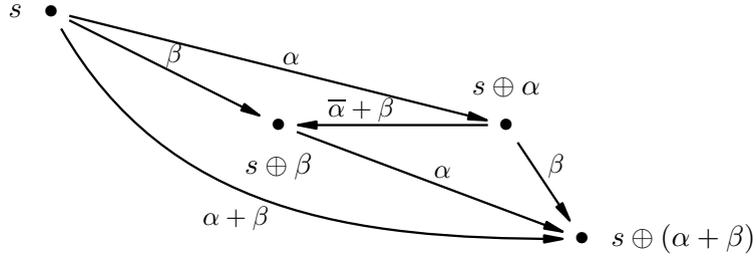

\medskip
In the domain of \textsc{pt}-nets, reversing
is introduced structurally rather than behaviourally,
by adding reverses
at the level of actions:

\medskip
A \emph{split reverse} of
$\STS$ is a step
transition system
$\STS^\splitrev=(S,T\uplus T',\to',s_0)$
satisfying \textit{SEQ} and  such that $T\cap\noidx(T')=\es$ and
$\noidx(\STS^\splitrev)=\STS^\rev$, where
$\noidx(\STS^\splitrev)
 =
 (S,T\cup\noidx(T'),\{(s,\noidx(\alpha),s')
 \mid (s, \alpha ,s')\in\to'\},s_0)$
is the step transition system obtained from $\STS$ by replacing
each occurrence of an indexed reverse action $\reverse{a}_{\IDX{\Idx}}$
by $\reverse{a}$.
That is, $\to'$ introduces split reverses allowing
one or more reverses of a step, possibly using different reverses
of the same action when reversing a step that contains its multiple
copies.
\begin{itemize}
\item
    A \textsc{pt}-net $N$ with \emph{reverses} is such that,
    for each original action $a$, there
    is a reverse action $\reverse{a}$
    such that $\eff_N(\reverse{a})=-\eff_N(a)$.

\item
    A \textsc{pt}-net $N$ with \emph{strict reverses} is such that,
    for each original action~$a$, there is
    a reverse action $\reverse{a}$
    such that $\PRE_N(\reverse{a})=\POST_N(a)$ and $\POST_N(\reverse{a})=\PRE_N(a)$.

\item
    A \textsc{pt}-net $N$ with \emph{split reverses} is such that,
    for each original action $a$, there
    is at least one indexed reverse action $\reverse{a}_{\IDX{\Idx}}$
    such that $\eff_N(\reverse{a}_{\IDX{\Idx}})=-\eff_N(a)$.
\end{itemize}

A key problem which then arises is that of characterising
relationships between
statically defined reversing of \textsc{pt}-nets and
the behavioural reversing of their concurrent reachability graphs.
In the rest of this paper, we will address this problem by
providing both negative and positive results.
First, however, we show basic properties of the reversed step transition systems.
In particular, that all such step transition systems are
\textsc{cest}-systems, and that the solvability of a reversed step transition system
implies the solvability of the original step transition system.

\begin{theorem}
\label{prop-new-iii}
    Let $\STS$ be a \textsc{cest}-system,
    and $\STS^\splitrev$ be any of its split reverses.
\begin{enumerate}
\item
    $\STS\bt \STS^\setrev\bt\STS^\rev\bt \STS^\mixrev$
    and $\STS\bt \STS^\splitrev$.
\item
    $\STS^\mixrev$,
    $\STS^\setrev$, $\STS^\rev$, and
    $\STS^\splitrev$ are \textsc{cest}-systems.

\item
    If any step transition system among
    $\STS^\mixrev$,
    $\STS^\setrev$, $\STS^\rev$, and
    $\STS^\splitrev$ is solvable, then
    $\STS$ is also solvable.
\end{enumerate}
\end{theorem}
\begin{proof}
    Let $\STS=(S,T,\to,s_0)$ and
    $\STS'$ be any step transition system
    among $\STS^\mixrev$,
    $\STS^\setrev$, $\STS^\rev$, and
    $\STS^\splitrev$.
    We start with an auxiliary result.

\begin{lemma}
\label{lem-ood}
    Let $\alpha,\beta,\gamma,\delta\in\mult(T)$.
\begin{enumerate}
\item
    $s\xright{\alpha}_{\STS^\mixrev}s'$
    iff
    $s\xright{\alpha}_{\STS^\rev} s'$
    iff
    $s\xright{\alpha}_{\STS^\splitrev} s'$
    iff
    $s\xright{\alpha}_\STS s'$.
\item
    $s\xright{\ralpha}_{\STS^\mixrev}s'$
    iff
    $s\xright{\ralpha}_{\STS^\rev}s'$.
\end{enumerate}
\end{lemma}
\begin{proof}[Lemma~\ref{lem-ood}]
(1)
    The second and third equivalences are obvious, so we only show the
    first one.

($\Longrightarrow$)
    Suppose that
    $s\xright{\alpha}_{\STS^\mixrev}s'$.
    Then, by the definition
    of $\STS^\mixrev$, there is $r\in S$ such that
    $r\xright{\es+\alpha}_\STS$
    and
    $(s=)r\oplus\es\xright{\reverse{\es}+\alpha}_{\STS^\mixrev}r\oplus\alpha(=s')$.
    By
    Proposition~\ref{prop-zt}(2), $s=r$.
    Hence, by Proposition~\ref{prop-zt}(2), $r\oplus \alpha=s\oplus \alpha=s'$.
    As a result,
    $s\xright{\alpha}_\STS s'$.

($\Longleftarrow$)
    Suppose that
    $s\xright{\alpha}_\STS s'$.
    Then $s\xright{\es+\alpha}_\STS$ and so,
    by the definition
    of $\STS^\mixrev$,
    $s\oplus\es\xright{\reverse{\es}+\alpha}_{\STS^\mixrev} s\oplus\alpha$.
    By Proposition~\ref{prop-zt}(1), $s'=s\oplus\alpha$, and, by
    Proposition~\ref{prop-zt}(2), $s=s\oplus\es$.
    Hence
    $s\xright{\alpha}_{\STS^\mixrev} s'$.

(2)
($\Longrightarrow$)
    Suppose that
    $s\xright{\ralpha}_{\STS^\mixrev}s'$.
    Then, by the definition
    of $\STS^\mixrev$, there is $r\in S$ such that
    $(s=)r\oplus\alpha\xright{\ralpha+\es}_{\STS^\mixrev}r\oplus\es(=s')$ and
    $r\xright{\alpha+\es}_\STS $.
    By Proposition~\ref{prop-zt}(2), $s'=r$.
    Hence
    $s'\xright{\alpha}_\STS s$. Thus, by the definition of $\STS^\rev$,
    $s\xright{\ralpha}_{\STS^\rev}s'$.

($\Longleftarrow$)
    Suppose that $s\xright{\ralpha}_{\STS^\rev}s'$.
    Then, by the definition
    of $\STS^\rev$, $s'\xright{\alpha+\es}_\STS s$.
    Hence, by definition of $\STS^\mixrev$,
    $s'\oplus\alpha\xright{\ralpha+\es}_{\STS^\mixrev}s'\oplus\es$.
    By Proposition~\ref{prop-zt}(1), $s=s'\oplus\alpha$, and, by
    Proposition~\ref{prop-zt}(2), $s'=s'\oplus\es$.
    Hence $s\xright{\alpha}_{\STS^\mixrev} s'$.
\end{proof}

(1)
    Follows directly from the definitions and
    Lemma~\ref{lem-ood}(1,2).

(2)
    We discuss in turn the four properties
    defining \textsc{cest}-systems.

(\textit{EL} and \textit{REA})
    Follow directly from part (1) and the fact that $\STS$ satisfies
    \textit{EL} and \textit{REA}.

(\textit{SEQ})
    For $\STS^\setrev$, $\STS^\rev$, and $\STS^\splitrev$,
    \textit{SEQ} holds directly from the definitions.
    To show \textit{SEQ} for $\STS^\mixrev$, suppose that:
\[
    s\xright{\alpha_1+\alpha_2+\beta_1+\beta_2}_\STS
    ~\mbox{and}~
    s\oplus(\alpha_1+\alpha_2)
     \xright{\ralpha_1+
                 \ralpha_2+\beta_1+\beta_2}_{\STS^\mixrev}
     s\oplus(\beta_1+\beta_2)\;.
\]
    Then, by \textit{SEQ} for $\STS$,
    we have
    $s\oplus \alpha_2\xright{\alpha_1+\beta_1}_\STS$
    and
    $s\oplus \beta_1\xright{\alpha_2+\beta_2}_\STS$.
    Hence, by the definition of $\STS^\mixrev$,
\[
\begin{array}{lll}
    (s\oplus\alpha_2)\oplus\alpha_1
    & \xright{\ralpha_1+\beta_1}_{\STS^\mixrev}
    & (s\oplus\alpha_2)\oplus\beta_1
\\
    (s\oplus \beta_1)\oplus\alpha_2
    & \xright{\ralpha_2+\beta_2}_{\STS^\mixrev}
    & (s\oplus \beta_1)\oplus\beta_2\;.
\end{array}
\]
    Moreover, by Proposition~\ref{prop-ffffp}, we have:
\[
\begin{array}{l@{~}c@{~}l@{~}c@{~}l}
    s\oplus(\alpha_2+\alpha_1)
    &=&
    (s\oplus\alpha_2)\oplus\alpha_1
\\
    (s\oplus \beta_1)\oplus\beta_2
    &=&
    s\oplus(\beta_1+\beta_2)
\\
    (s\oplus\alpha_2)\oplus\beta_1
    &=&
    s\oplus(\alpha_2 +\beta_1)
    &=&
    (s\oplus \beta_1)\oplus\alpha_2 \;.
\end{array}
\]
    Hence,
    $
    s\oplus(\alpha_1+\alpha_2)
    \xright{\ralpha_1+\beta_1}_{\STS^\mixrev}
    s\oplus(\alpha_2+\beta_1)
    \xright{\ralpha_2+\beta_2}_{\STS^\mixrev}
    s\oplus (\beta_1+\beta_2)
    $.

\medskip
(\textit{CE})
    We first observe that
    $s\xright{\reverse{a}}_{\STS^\mixrev} s'$ implies
    $s'\xright{\reverse{a}}_{\STS^\mixrev} s$, by Lemma~\ref{lem-ood}
    and the definition of $\STS^\rev$ \textit{(*)}.

    We have already demonstrated that \textit{SEQ}
    holds for $\STS'$.
    Hence, by Propositions~\ref{prop-x1}, it suffices
    to show that \textit{CE} holds for $(\STS')^\seq$.

    By Propositions~\ref{prop-x1}, we have that $\STS^\seq$
    satisfies \textit{CE}.
    Moreover,
    by Lemma~\ref{lem-ood}(1) as well as the  definition
    of $\STS'$ and  \textit{(*)},
    $(\STS')^\seq$ can be derived by a successive application
    of the construction from the formulation of
    Proposition~\ref{prop-x2} (once for each reverse action
    and indexed reverse action).
    Hence, by Propositions~\ref{prop-x2},
    $(\STS')^\seq$ satisfies \textit{CE}.

(3)
    Let
    $N'=(P,T',F,M_0)$ be a \textsc{pt}-net such that
    $\STS' \simeq_\psi\CRG_{N'}$.
    We will show that $\STS\simeq_\psi \CRG_N$, where $N=N'|_T$.
    Note that the enabling and firing of steps
    over $T$ is exactly the same in both $N$ and $N'$
    \textit{(*)}.

    We first observe that $\psi(s_0)=M_0$.
    Suppose then that
    $s\in S$ and $\psi(s)\in\reach_N$.
    To show that the executions of steps are
    preserved by $\psi$ in both directions,
    we consider two cases for $\alpha\in\mult(T)$.

\medskip
\emph{Case 1:}
    $s\xright{\alpha}_\STS s'$.
    Then, by  part (1),
    $s\xright{\alpha}_{\STS'} s'$.
    Hence, by $\STS' \simeq_\psi \CRG_{N'}$,
    we have $\psi(s)\xright{\alpha}_{N'} \psi(s')$.
    Thus, by \textit{(*)},
    $\psi(s)\xright{\alpha}_N \psi(s')$.

\emph{Case 2:}
    $\psi(s)\xright{\alpha}_N M$.
    Then, by \textit{(*)},
    $\psi(s)\xright{\alpha}_{N'} M$.
    Hence, by $\STS' \simeq_\psi \CRG_{N'}$,
    we have
    $M\in\psi(S)$
    and
    $s\xright{\alpha}_{\STS'} \psi^{-1}(M)$.
    Thus, by  Lemma~\ref{lem-ood}(1),
    $s\xright{\alpha}_\STS \psi^{-1}(M)$.
\end{proof}

\section{Multiset and set reversibility}
\label{sect-1}

The investigation of different notions of
step reversibility starts with a
straightforward but important negative result stating
that, in the domain of \textsc{pt}-nets,
the concept of direct reversibility
--- which directly generalises sequential reversibility
and should be considered as the preferred way
of reversing step transition systems ---
cannot handle
steps which are true multisets.

\begin{proposition}
\label{th-1}
    Let $\STS$ be a \textsc{cest}-system
    which is not a set transition system.
    Then $\STS^\rev$ is not solvable.
\end{proposition}

\begin{proof}
    [Figure~\ref{fig-ex1}($a$) illustrates the idea of the
    proof.]
    Let $\STS=(S,T,\to,s_0)$. Suppose that
    $\STS^\rev$ is solvable.
    Then there is a
    \textsc{pt}-net $N$ such that $\STS^\rev\simeq_\psi\CRG_N$
    \textit{(*)}.
    As $\STS$ is not a set transition system, there are
    $v\in S$ and $\alpha\in\mult(T)$ such that
    $v\xright{\alpha}_\STS$ and
    $(aa)\leq\alpha$, for some $a\in T$.

\begin{figure}[h]
\begin{center}
($a$)
\begin{tikzpicture}[node distance=1.3cm,>=arrow30,%
     line width=0.3mm,scale=1.0,bend angle=45]
     \tikzstyle{box}=[draw,regular polygon,thick,%
     regular polygon sides=4,minimum size=22mm, inner sep = -3pt]
\node (v) [label=left:$v$] at (0,0) {$\bullet$};
\node (q) [label=$q$] at (2,0) {$\bullet$};
\node (w) [label=right:$w$] at (4,0) {$\bullet$};

\draw[-arrow30] (v) to node [auto,swap,inner sep=2pt] {\small $a$} (q);
\draw[-arrow30] (q) to node [auto,swap,inner sep=2pt] {\small $a$} (w);
\draw[-arrow30] (w) to [out=150,in=30] node [auto,swap,inner sep=1pt]
{\small $\reverse{a}$} (q);
\draw[-arrow30] (q) to [out=150,in=30] node [auto,swap,inner sep=1pt]
{\small $\reverse{a}$} (v);

\draw[-arrow30] (v) to [out=45,in=135] node [auto,inner sep=2pt]
{\small $(aa)$} (w);
\draw[-arrow30] (w) to [out=-120,in=-60] node [auto,inner sep=2pt]
{\small $(\reverse{aa})$} (v);
\draw[-arrow30] (q) to [out=-30,in=-150, looseness=8]
node [auto,inner sep=2pt] {\small $(a\reverse{a})$} (q);
\end{tikzpicture}
~~~~~~~~~~~~~~~
\begin{tikzpicture}[scale=0.9]

\node[circle,draw,minimum size=0.7cm] (p2) at (3,1) {$\bullet$};
\node[circle,draw,minimum size=0.7cm] (p1) at (1,1) { };

\node[draw,minimum size=0.6cm] (a1) at (0,0){$a$};
\node[draw,minimum size=0.6cm] (a2) at (0,2){$\reverse{a}$};
\node[draw,minimum size=0.6cm] (b2) at (2,2){$\reverse{b}$};
\node[draw,minimum size=0.6cm] (b1) at (2,0){$b$};

\draw[-latex] (p1) to node [auto,inner sep=1pt] {} (a1);
\draw[-latex] (a1) to node [auto,inner sep=1pt] {} (p1);
\draw[-latex] (p1) to node [auto,inner sep=1pt] {} (a2);
\draw[-latex] (a2) to node [auto,inner sep=1pt] {} (p1);

\draw[-latex] (p2) to node [auto,inner sep=1pt] {} (b1);
\draw[-latex] (b1) to node [auto,inner sep=1pt] {} (p2);
\draw[-latex] (p2) to node [auto,inner sep=1pt] {} (b2);
\draw[-latex] (b2) to node [auto,inner sep=1pt] {} (p2);

\draw[-latex] (b1) to node [auto,inner sep=1pt] {} (p1);
\draw[-latex] (p1) to node [auto,inner sep=1pt] {} (b2);
\end{tikzpicture}
($b$)
\end{center}\vspace*{-3mm}
\caption{An illustration of the proof of
Proposition~\ref{th-1} ($a$), and \textsc{pt}-net
generating concurrent reachability graph which is not
step-finite ($b$).
\label{fig-ex1}}
\end{figure}
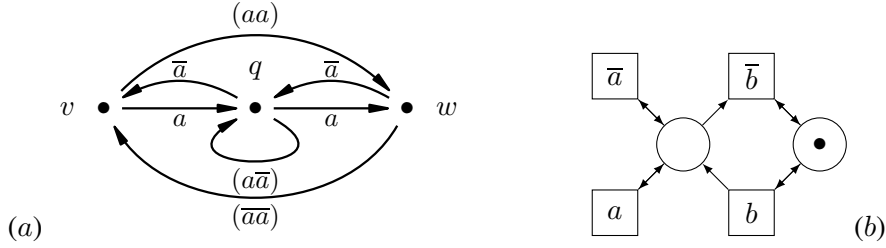

    By \textit{SEQ} for $\STS$ and Theorem~\ref{prop-new-iii}(1),
    there are $w,q\in S$ such that
    $v\xright{(aa)}_{\STS^\rev} w$ and $v\xright{a}_{\STS^\rev} q$
    \textit{(**)}.
    Hence,
    by the definition of $\STS^\rev$,
    $w\xright{(\reverse{a}\reverse{a})}_{\STS^\rev} v$
    \textit{(***)}.

    Let $M_s=\psi (s)$, for $s\in\{v,w,q\}$.
    By the definition of $\STS^\rev$ and \textit{(*)},
    the step $\beta=(a\reverse{a})$ is not enabled at $M_q$.
    Hence, there is a place $p$ of $N$ such that
    $M_q(p)<\PRE_N(\beta)(p)$ \textit{(\dag)}.
    On the other hand, by \textit{(**)} and \textit{(***)},
    we have:
\[
    \PRE_N(aa)
    \leq
    M_v
    ~~~~~~~
    \PRE_N(\reverse{a}\reverse{a})
    \leq
    M_w
    ~~~~~~~
    M_w
    =
    M_v+ \eff_N(aa)
    ~~~~~~~
    M_q
    =
    M_v+\eff_N(a)\;.
\]
    Thus
    $\PRE_N(\beta) +\PRE_N(\beta)
     =\PRE_N(aa\reverse{a}\reverse{a})\leq M_v + M_w
     =  M_v + M_v +  \eff_N(aa)= M_q+M_q$,
    yielding a contradiction with \textit{(\dag)}.
\end{proof}

In view of Proposition~\ref{th-1}, when facing the problem of
implementing a reverse of non-set step transition system $\STS$
using \textsc{pt}-nets, one may consider set reversibility
based on $\STS^\setrev$,
or mixed
reversibility based on $\STS^\mixrev$.\footnote
{
    We will discuss split reversibility separately
    in Section~\ref{sect-split}.
}

Among these two options, one might prefer $\STS^\setrev$ to
$\STS^\mixrev$ as the latter introduces steps containing both the original
and reverse actions.
However, as the next example shows,
it not always possible to `replace' a mixed reversibility
solution by a set reversibility solution.

\begin{example}
\label{ex-iidid}
    Let us consider a \textsc{cest}-system
    $\STS=(\{s_0,s_1,\dots\},\{a,b\},\to,s_0)$
    such that:
\[
    s_i\xright{a^j}_\STS s_i
    ~~~~~~~~\textrm{and}~~~~~~~~
    s_i\xright{b+a^j}_\STS s_{i+1}
    ~~ \textrm{for all}~i\geq 0~\textrm{and}~j\leq i \;.
\]
    It is straightforward to see
    that $\STS^\mixrev$ is solvable by the
    \textsc{pt}-net shown in Figure~\ref{fig-ex1}($b$).
    However,
    $\STS^\setrev$ is \emph{not} solvable by any \textsc{pt}-net.
    If such a \textsc{pt}-net $N$ existed, then
    it would have distinct reachable markings $M_0,M_1,\dots$
    satisfying, for every $i\geq 0$:
\[
    M_i\xright{b}_N M_{i+1} ~\textit{(*)}
    ~~~ ~~~ ~~~~
    M_i\xright{a^i}_N M_i ~ \textit{(**)}
    ~~~ ~~~ ~~~~
    M_i\xright{\reverse{a}}_N M_i ~ \textit{(***)}
    ~~~ ~~~ ~~~~
    \neg M_i\xright{(a\reverse{a})}_N   ~ \textit{(\dag)}\;.
\]
    We now observe that $M_0\leq M_1\leq \cdots$
    due to \textit{(*)}.
    Hence, there is
    a place $p$ such that
    $\PRE_N(a\reverse{a})(p)>M_0(p)=M_1(p)=\cdots$ \textit{($\ddag$)},
    due to
    \textit{($\dag$)} and the finiteness of
    $N$.
    On the other hand,
    $\PRE_N(\reverse{a})(p)\leq M_0(p)=M_1(p)=\cdots$
    due to
    \textit{(***)}, and
    $\PRE_N(a)(p)=0$
    due to \textit{(**)} and \textit{($\ddag$)}.
    As a result, $\PRE_N(a\reverse{a})(p)\leq M_0(p)$,
    yielding a contradiction with \textit{($\ddag$)}.
\hfill$\diamondsuit$
\end{example}

Example~\ref{ex-iidid} demonstrated that there are
step transition systems which can be treated using mixed
reversibility, but not using set reversibility.
What is more, the example worked because the step transition
system considered was not step-finite.
As the next result shows, that was the only reason why
set reversibility failed to hold.

\begin{theorem}
\label{pr-mix2set}
    Let $\STS$ be a \textsc{cest}-system such
    that $\STS^\mixrev$ is solvable.
    Then $\STS^\setrev$ is solvable if and only if
    $\STS$ is step-finite.
\end{theorem}
\begin{proof}
    Let $\STS=(S,T,\to,s_0)$.

($\Longrightarrow$)
    Suppose that $\STS^\setrev$ is solvable by a \textsc{pt}-net
    $N=(P,T\cup\rT,F,M_0)$, and that $\STS$ is not step-finite.
    By the finiteness of $P$ and $T$ as well as \textit{SEQ} for $\STS$,
    there is $a\in T$ and reachable markings $M_1\leq M_2\leq\dots$
    such that $M_i\xright{a^i}_N$,
    for every $i\geq 1$.
    Hence, by \textit{SEQ} for $\CRG_N$,
    there is a marking $M'_i$
    such that
    $M_i\xright{a}_N M'_i$
    and
    $M'_i\xright{a^{i-1}}_N$ \textit{($*$)},
    for every $i\geq 1$.
    As a result, $M'_i \xright{a}_N$
    and $ M'_i \xright{\reverse{a}}_N$ \textit{($**$)}, for every $i\geq 2$.

    We now observe that $(M=)M'_{m+2} \xright{(a\reverse{a})}_N$,
    where $m= max \{ F(p,\reverse{a})\mid p \in P\}$.
    Indeed, otherwise there is $p\in P$ such that
    $M(p)< F(p,a)+F(p,\reverse{a})\leq F(p,a)+m$ \textit{($\dag$)}.
    On the other hand, by \textit{($**$)},
    $M(p)\geq F(p,a)$ and $M(p)\geq F(p,\reverse{a})$.
    Hence, it must be the case that
    $F(p,a)>0$.
    Thus, by \textit{($*$)},
    $M(p)\geq (m+1)\cdot F(a,p)=m +F(a,p)$, contradicting \textit{($\dag$)}.
    As a result,
    $M\xright{(a\reverse{a})}_N$, yielding a contradiction with our
    initial assumption.

\medskip

($\Longleftarrow$)
    If $\STS$ is step-finite, then there is
    $k\geq 1$ such that
    $|\alpha|\leq k$, whenever $s\xright{\alpha}_\STS$.
    Moreover, since $\STS^\mixrev$ is solvable,
    there exists a \textsc{pt}-net $N=(P,T\cup\rT,F,M_0)$
    such that $\STS^\mixrev\simeq_\psi \CRG_N$.
    We then modify $N$,  by adding to $P$
    a set of fresh places
    $P'=\{p_{ab}\mid a\in T\wedge b\in\rT\}$. Each $p_{ab}$
    is such that $M_0(p_{ab})=k$ and has
    four non-zero connections,
    $F(a,p_{ab})=F(p_{ab},a)=1$
    and
    $F(b,p_{ab})=F(p_{ab},b)=k$.
    For the resulting
    \textsc{pt}-net $N'$, we have
    $\STS^\setrev \simeq_{\psi'} \CRG_{N'}$,
    where
    $\psi'(s)=\psi(s)+\sum_{p\in P'}p^k$,
    for every $s\in S$.
\end{proof}

We have therefore obtained a full characterisation of
step transition systems for which mixed reversibility solutions
can be replaced by set reversibility solutions.
In addition, the second part of
the proof of Theorem~\ref{pr-mix2set} provides
a straightforward construction achieving this.

A direct corollary of the last result is that
for a set step transition  system it is always possible to replace
a mixed reversibility solution
by a set reversibility solution.

\begin{theorem}
\label{cor-ddd}
    Let $\STS$ be a set \textsc{cest}-system.
    If $\STS^\mixrev$ is solvable, then
    $\STS^\rev$ is also solvable.
\end{theorem}
\begin{proof}
    As a set \textsc{cest}-system, $\STS$ is step-finite and
    $\STS^\rev=\STS^\setrev$.
	Hence the result follows from Theorem~\ref{pr-mix2set}.
\end{proof}

A concluding observation is that all three
versions of reversibility which do not involve splitting are
worthy of investigation.

\section{Mixed reversibility}
\label{sect-2}

In this section, we consider the problem of deciding whether
the mixed reverse $\STS^\mixrev$ of a solvable step transition system
$\STS$ is also solvable.
A specific concern we implicitly address is
the size of $\STS^\mixrev$ which (in the finite case)
can be exponentially larger than that of $\STS$.
The aim is therefore to avoid dealing directly with $\STS^\mixrev$.
As shown below, this is possible as
the checking of feasibility of mixed reversing can  be
replaced by checking the solvability of
the original transition system,
and the solvability of its reverse.

\medskip
Throughout this section we make the following assumptions:
\begin{itemize}
\item
    $\STS=(S,T,\to,s_0)$ is a \textsc{cest}-system and
    $R$ is a home cover of $\STS$.
\item
    $\rSTS=(S,\rT,\{(s',\reverse{\alpha},s)\mid s\xright{\alpha}_\STS s'\},R)$
    is a step transition system
    with multiple initial states.
\item
    $\rSTS_r=(S_r,\rT,\to_r,r)$ is a step transition system such that
    $r\in R$,
    $S_r=\{s\in S\mid r\in\pred_\STS(s)\}$,
    and
    $\to_r\,=\,\to\cap\,(S_r\times\mult(T)\times S_r)$.
\end{itemize}
That is, $\rSTS$ is obtained by reversing each transition of $\STS$,
and considering all the states in the home cover $R$ as the initial states.

\begin{proposition}
\label{prop-dhhd}
    Let $r\in R$.
\begin{enumerate}
\item
    $\rSTS_r$ is  a \textsc{cest}-system.
\item
    $s_0 \in \bigcap_{s\in S_r}\pred_\STS(s)$.
\item
    $S=\bigcup_{r\in R}S_r$.
\end{enumerate}
\end{proposition}
\begin{proof}
(1)
    The only non-trivial property to show is \textit{CE}.
    For every $\mathcal A$-vector $\alpha$ with support $\rT$, let
    $\widehat{\alpha}$ be the $\mathcal A$-vector with support $T$
    such that
    $\widehat{\alpha}(a)=-\alpha(\reverse{a})$, for every $a\in T$.

    We first observe that, for every $\pi\in\paths_{\rSTS_r}(s,s')$, there is
    $\pi'\in\paths_\STS(s,s')$ such that $\sign(\pi')=\widehat{\sign(\pi)}$
    \textit{(*)}.
    Hence, we also have that $\alpha\bowtie_{\rSTS_r}\beta$
    implies $\widehat{\alpha}\bowtie_\STS\widehat{\beta}$,
    for all $\mathcal A$-vectors $\alpha$ and $\beta$ with support $\rT$
    \textit{(**)}.
    Thus, $\rSTS_r$ satisfies \textit{CE} by \textit{(*)} and \textit{(**)}.\smallskip

\noindent (2)
    Follows from the fact that $\STS$ satisfies \textit{REA}.\smallskip

\noindent (3)
    Follows from the fact that $R$ is a home cover.
\end{proof}

\begin{theorem}
\label{th-symsyn}
    $\STS^\mixrev$ is solvable if and only if
    both
    $\STS$ and $\rSTS$
    are  solvable.
\end{theorem}
\begin{proof}
($\Longrightarrow$)
    By Theorem~\ref{prop-new-iii}(3), $\STS$ is solvable.
    To show that $\rSTS$ is solvable,
    suppose that $N=(P,T,F,M_0)$ is a \textsc{pt}-net such that
    $\STS^\mixrev \simeq_\psi\CRG_N$.
    We will show that $\rSTS_r\simeq_{\psi|_{S_r}} \CRG_{N_r}$,
    where, for every $r\in R$, $N_r$ is the \textsc{pt}-net $N|_{\rT}$
    with the initial marking set to $\psi(r)$.
    Note that the enabling and firing of steps
    over $\rT$ is exactly the same in both $N$ and $N_r$
    \textit{(*)}.

    We first observe that the initial states
    of $\rSTS_r$ and $\CRG_{N_r}$ are related by $\psi$.
    Suppose then that
    $s\in S_r$ is such that $\psi(s)\in\reach_{N_r}$.
    To show that the executions of steps are
    preserved by $\psi$ in both directions,
    we consider two cases, where $\alpha\in\mult(T)$.

\medskip
\emph{Case 1.1:}
    $s\xright{\ralpha}_{\rSTS_r} s'$.
    Then $s\xright{\ralpha}_{\STS^\rev}s'$
    and so, by Lemma~\ref{lem-ood}(2),
    $s\xright{\ralpha}_{\STS^\mixrev} s'$.
    Hence, by $\STS^\mixrev \simeq_\psi \CRG_N$, we have
    $\psi(s)\xright{\ralpha}_N \psi(s')$.
    Thus, by \textit{(*)},
    $\psi(s)\xright{\ralpha}_{N_r} \psi(s')$.

\emph{Case 1.2:}
    $\psi(s)\xright{\ralpha}_{N_r} M$.
    Then, by \textit{(*)},
    $\psi(s)\xright{\ralpha}_N M$.
    Hence, by $\STS^\mixrev \simeq_\psi \CRG_N$,
    we have $M\in\psi(S)$ and
    $s\xright{\ralpha}_{\STS^\mixrev} \psi^{-1}(M)$.
    Thus, by Lemma~\ref{lem-ood}(2),
    $s\xright{\ralpha}_{\STS^\rev} \psi^{-1}(M)$.
    Hence
    $s\xright{\ralpha}_{\rSTS_r} \psi^{-1}(M)$.

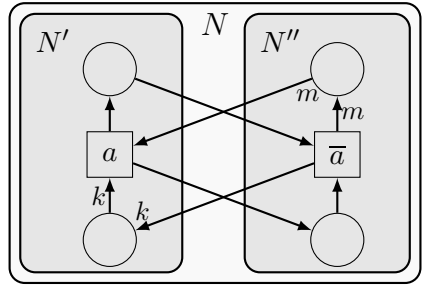
\begin{figure}[h]
\begin{center}
\begin{tikzpicture}[scale=0.75]
\node[] (n1) at(-1,4.5){$N'$};
\node[circle,draw,minimum size=0.7cm] (p1) at (0,4) {};
\node[circle,draw,minimum size=0.7cm] (p1p) at (0,1) {};
\node[] (x1) at(1,4.5){};
\node[draw,minimum size=0.6cm] (t1) at (0,2.5){$a$};
\node[] (n2) at(3,4.5){$N''$};
\node[circle,draw,minimum size=0.7cm] (p2) at (4,4) {};
\node[circle,draw,minimum size=0.7cm] (p2p) at (4,1) {};
\node[] (x2) at(5,4.5){};
\node[draw,minimum size=0.6cm] (t2) at (4,2.5){$\reverse{a}$};
\node[] (n3) at(1.85,4.8){\large $N$};
\draw[-latex,thick] (t1) to node [auto,inner sep=1pt] {} (p1);
\draw[-latex,thick] (p1p) to node [auto,inner sep=1pt] {\small $k$} (t1);
\draw[-latex,thick] (p2) to node [auto,inner sep=1pt] {} (t1);
\draw[-latex,thick] (t1) to node [auto,inner sep=1pt] {} (p2p);
\draw[-latex,thick] (t2) to node [auto,inner sep=1pt,xshift=-35,yshift=-0.3]
{\small $k$} (p1p);
\draw[-latex,thick] (p1) to node [auto,swap,inner sep=1pt,xshift=37,yshift=10]
{\small $m$} (t2);
\draw[-latex,thick] (p2p)to node [auto,inner sep=1pt] {} (t2);
\draw[-latex,thick] (t2) to node [auto,swap,inner sep=1pt] {\small $m$} (p2);

\begin{pgfonlayer}{background}
\node[fill=gray!5,thick,draw=black,rounded corners=2mm,inner sep=15pt]
(bg3) [fit = (n1) (n2) (p2p) (x2), inner sep=6pt] {};
\end{pgfonlayer}
\begin{pgfonlayer}{background}
\node[fill=gray!20,thick,draw=black,rounded corners=2mm]
(bg1) [fit = (n1) (x1) (p1p), inner sep=2pt] {};
\end{pgfonlayer}
\begin{pgfonlayer}{background}
\node[fill=gray!20,thick,draw=black,rounded corners=2mm]
(bg2) [fit = (n2) (x2) (p2p), inner sep=2pt] {};
\end{pgfonlayer}
\end{tikzpicture}
\end{center}\vspace*{-3mm}
\caption{
    An illustration of the second part of
    the proof of Theorem~\ref{th-symsyn}.
\label{fig-a2} }\vspace*{1mm}
\end{figure}

($\Longleftarrow$)
    Since $\STS$ is solvable, there is a \textsc{pt}-net
    $N'=(P',T,F',M'_0)$ such that $\STS\simeq_{\psi'}\CRG_{N'}$.
    (Note that $\psi'(s_0)=M'_0$.)
    Moreover, since $\rSTS$ is solvable,
    there is an umarked \textsc{pt}-net $N''=(P'',\rT,F'')$
    and a mapping $\psi'':S\to\mult(P'')$ such that
    $\rSTS_r\simeq_{\psi''|_{S_r}}\CRG_{N_r}$,
    where
    $N_r=(P'',\rT,F'',M_r)$ and $M_r=\psi''(r)$,
    for every $r\in R$.
    Clearly, we may assume that $P'\cap P''=\es$
    as the identities of places play no role in
    the solvability problems of $\STS$ and $\rSTS$.

\medskip
    Let
    $N =(P'\cup P'', T\cup\rT,F,M_0)$ be the \textsc{pt}-net
    with
    strict reverses (illustrated in Figure~\ref{fig-a2}) such that
    $M_0=M'_0\munion\psi''(s_0)=\psi'(s_0)\munion\psi''(s_0)$
    and, for every $a\in T$:
\begin{equation}
\label{eq-ppp1}
\begin{array}{l@{~}c@{~}l@{~}c@{~}l@{~~~~~~~~~~~~}l@{~}c@{~}l@{~}c@{~}l}
    \PRE _N(a)
    &=&
    \PRE _{N'}(a)
    &\munion&
    \POST_{N''}(\reverse{a})
    &
    \POST_N(a)
    &=&
    \POST_{N'}(a)
    &\munion&
    \PRE _{N''}(\reverse{a})
\\
    \PRE _N(\reverse{a})
    &=&
    \PRE _{N''}(\reverse{a})
    &\munion&
    \POST_{N'}(a)
    &
    \POST_N(\reverse{a})
    &=
    &\POST_{N''}(\reverse{a})
    &\munion&
    \PRE _{N'}(a) \;.
\end{array}
\end{equation}

    Let $\psi$ be a mapping with the domain $S$ which, for
    every $s\in S$, returns $\psi'(s)\munion\psi''(s)$.
    Note that $\psi$ is well-defined due to Lemma~\ref{prop-dhhd}(3)
    and   $\psi(s_0)=M_0$.

\begin{lemma}
\label{lemma-77d}
    Let $\STS'$ be $\CRG_N$ with all the
    transitions labelled by  steps of the form $\alpha+\rbeta$, for
    $\alpha,\beta\neq\es$, deleted.
\begin{enumerate}
\item
    $\STS^\rev\simeq_\psi \STS'$.
\item
    $\STS'$ satisfies \textit{REA}.
\item
    $\psi(s\oplus\alpha)=\psi(s)+\eff_N(\alpha)$, for all $s\xright{\alpha}_\STS$.
\end{enumerate}
\end{lemma}
\begin{proof}
[Lemma~\ref{lemma-77d}]
(1)
    We observe that the initial states of $\STS^\rev$ and $\STS'$
    are related by $\psi$.
    Suppose now that $s\in S$ and $\psi(s)\in\reach_N$.
    To show that the executions of steps are
    preserved by $\psi$ in both directions,
    we consider four cases, where $\alpha\in\mult(T)$.

\medskip
\emph{Case 2.1:}
    $s\xright{\alpha}_{\STS^\rev} s'$.
    Then, by $\STS\simeq_{\psi'}\CRG_{N'}$, we have
    $\psi'(s)\xright{\alpha}_{N'}\psi'(s')$ and
    $\psi'(s)\geq\PRE_{N'}(\alpha)$.
    Moreover, $s'\xright{\ralpha}_{\STS^\rev} s$. Hence,
    by Lemma~\ref{prop-dhhd}(3),
    there is $r\in R$ such that $s'\xright{\ralpha}_{\rSTS_r} s$.
    Thus, by $\rSTS_r\simeq_{\psi''|_{S_r}}\CRG_{N_r}$, we have
    $\psi''(s')\xright{\ralpha}_{N''}\psi''(s)$ and
    $\psi''(s)\geq \POST_{N''}(\ralpha)$. Hence, by Eq.(\ref{eq-ppp1}):
\[
    \psi(s)
    =
    (\psi'(s)\munion\psi''(s))
    \geq
    (\PRE_{N'}(\alpha)\munion\POST_{N''}(\ralpha))
    =
    \PRE _N(\alpha)\;.
\]
    As a result, $\psi(s)\xright{\alpha}_N\psi(s)+\eff_N(\alpha)$.
    Hence $\psi(s)\xright{\alpha}_N\psi(s')$ as we have, by Eq.(\ref{eq-ppp1}):
\[
\begin{array}{l@{~}c@{~}l}
    \psi(s)+\eff_N(\alpha)
    &=&
    (\psi'(s)\munion\psi''(s))+\POST_N(\alpha)-\PRE_N(\alpha)
    \\
    &=&
    (\psi'(s)\munion\psi''(s))+(\POST_{N'}(\alpha)\munion\PRE_{N''}(\ralpha))
    -(\PRE_{N'}(\alpha)\munion\POST_{N''}(\ralpha))
    \\
    &=&
    (\psi'(s)+\eff_{N'}(\alpha))
    \munion
    (\psi''(s)-\eff_{N''}(\ralpha))
    \\
    &=&
    \psi'(s')\munion\psi''(s')=\psi(s')\;.
\end{array}
\]

\emph{Case 2.2:}
    $s\xright{\ralpha}_{\STS^\rev} s'$.
    Then $s'\xright{\alpha}_{\STS^\rev} s$ and so, by Case~2.1,
    $\psi(s')\xright{\alpha}_N\psi(s)$.
    Hence, since $N$ is a \textsc{pt}-net with strict
    reverses, $\psi(s)\xright{\ralpha}_N\psi(s')$.

\medskip
\emph{Case 2.3:}
    $\psi(s)\xright{\alpha}_N M$.
    Then, by Eq.(\ref{eq-ppp1}), we have:
\[
\begin{array}{c}
    \psi'(s)\munion\psi''(s)=\psi(s)\geq \PRE_N(\alpha)
    = \PRE_{N'}(\alpha)\munion\POST_{N''}(\ralpha)
\\
    M=\psi(s)+\eff_N(\alpha)=(\psi'(s)\munion\psi''(s))
    +(\POST_{N'}(\alpha)\munion\PRE_{N''}(\ralpha))
    -(\PRE_{N'}(\alpha)\munion\POST_{N''}(\ralpha)).
\end{array}
\]
    Hence, by $P'\cap P''=\es$,
    $\psi'(s)\geq \PRE_{N'}(\alpha)$ and
    $\psi''(s)\geq \POST_{N''}(\ralpha)$
    as well as:
\[
    \proj{M}{P'}  = \psi'(s)  + \eff_{N'}(\alpha)
    ~~~~~~\textrm{and}~~~~~~
    \proj{M}{P''} = \psi''(s) - \eff_{N''}(\ralpha)\;.
\]
    Thus $\psi'(s)\xright{\alpha}_{N'}\proj{M}{P'}$.
    Hence, by $\STS\simeq_{\psi'}\CRG_{N'}$,
    we obtain $\proj{M}{P'}\in \psi'(S)$ and
    $s\xright{\alpha}_{\STS^\rev} s'$, where $\psi'(s')=\proj{M}{P'}$.
    We still need to show that $\psi(s')=M$.
    This follows from
    $\psi''(s')=\proj{M}{P''}$.
    Indeed, we have
    $s' \xright{\ralpha}_{\STS^\rev} s$ and so,
    by Lemma~\ref{prop-dhhd}(3), there is $r\in R$ such that
    $s'\in S_r$. Now, by $\rSTS_r\simeq_{\psi''|_{S_r}}\CRG_{N_r}$,
    $\psi''(s')\xright{\ralpha}_{N''}\psi''(s)$, which
    means that
    $\psi''(s')=\psi''(s)-\eff_{N''}(\ralpha)=\proj{M}{P''}$.

\medskip
\emph{Case 2.4:}
    $\psi(s)\xright{\ralpha}_N M$.
    Then, by Eq.(\ref{eq-ppp1}), we have:
\[
\begin{array}{c}
    \psi'(s)\munion\psi''(s)=\psi(s)\geq \PRE_N(\ralpha)
    = \PRE_{N''}(\ralpha)\munion\POST_{N'}(\alpha)
\\
    M=(\psi'(s)\munion\psi''(s))
    +(\POST_{N''}(\ralpha)\munion\PRE_{N'}(\alpha))
    -(\PRE_{N''}(\ralpha)\munion\POST_{N'}(\alpha))\;.
\end{array}
\]
    Hence, by $P'\cap P''=\es$,
    $\psi'(s)\geq \POST_{N'}(\alpha)$ and
    $\psi''(s)\geq \PRE_{N''}(\ralpha)$ as well as:
\[
    \proj{M}{P'}  = \psi'(s)  -\eff_{N'}(\alpha)
    ~~~~~~\textrm{and}~~~~~~
    \proj{M}{P''} = \psi''(s) +\eff_{N''}(\ralpha)\;.
\]
    Thus $\psi''(s)\xright{\ralpha}_{N''}\proj{M}{P''}$.
    Hence, by Lemma~\ref{prop-dhhd}(3), there is $r\in R$ such that
    $s\in S_r$. Thus, by $\rSTS_r\simeq_{\psi''|_{S_r}}\CRG_{N_r}$,
    $\proj{M}{P''}\in \psi''(S)$ and
    $s\xright{\ralpha}_{\STS^\rev} s'$, where $\psi''(s')=\proj{M}{P''}$.
    We still need to show that $\psi(s)=M$. This   follows from
    $\psi'(s')=\proj{M}{P'}$.
    Indeed, we have
    $s' \xright{\alpha}_{\STS^\rev} s$ and so,
    by $\STS\simeq_{\psi'}\CRG_{N'}$,
    we obtain $\psi'(s')\xright{\alpha}_{N'}\psi'(s)$, which means
    that
    $\psi'(s')=\psi'(s)-\eff_{N'}(\alpha)=\proj{M}{P'}$.

(2)
    The modification of $\CRG_N$ does not
    produce unreachable states since $\CRG_N$ satisfies \textit{SEQ}.

(3)
    Follows from part (1) and the forward determinism of $\STS$ and $\CRG_N$.
\end{proof}

    Returning to the proof of $\STS^\mixrev \simeq_\psi\CRG_{N}$,
    suppose that $s\in S$ is such that $\psi(s)\in\reach_N$
    and consider two cases, where $\alpha,\beta\in\mult(T)$.

\medskip
\emph{Case 3.1:}
    $s\xright{\alpha+\beta}_\STS$ and
    $s\oplus\alpha\xright{\ralpha+\beta}_{\STS^\mixrev} s\oplus\beta$.
    Then we have
    $s\xright{\alpha+\beta}_{\STS^\rev}$ as well as:
\[
    s\xright{\alpha}_\STS s\oplus\alpha
    ~~~~~~~~~~~
    s\xright{\beta}_\STS s\oplus\beta
    ~~~~~~~~~~~
    s\xright{\alpha}_{\STS^\rev} s\oplus\alpha
    ~~~~~~~~~~~
    s\xright{\beta}_{\STS^\rev} s\oplus\beta\;.
\]
    Hence, by Lemma~\ref{lemma-77d}(1,3),
    we have:
\[
    \psi(s)\xright{\alpha+\beta}_N
    ~~~~~
    \psi(s)\xright{\alpha}_N \psi(s\oplus\alpha)=\psi(s)+\eff_N(\alpha)
    ~~~~~
    \psi(s)\xright{\beta}_N \psi(s\oplus\beta)=\psi(s)+\eff_N(\beta)\;.
\]
    Thus
    $\psi(s)\geq \PRE_N(\alpha+\beta)$, and so
    $
    \psi(s)+\eff_N(\alpha)
    \geq
    \PRE_N(\alpha+\beta)+\eff_N(\alpha)
    =
    \PRE_N(\ralpha+\beta)
    $
    due to   Eq.(\ref{eq-ppp1}).
    Hence, again by Eq.(\ref{eq-ppp1}):
\[
    \psi(s\oplus \alpha)=\psi(s)+\eff_N(\alpha)
     \xright{\ralpha+\beta}_N
     \psi(s)+\eff_N(\alpha)+\eff_N(\ralpha+\beta)
     =
     \psi(s)+\eff_N(\beta)=\psi(s\oplus \beta)\;.
\]

\emph{Case 3.2:}
    $\psi(s)\xright{\ralpha+\beta}_N M$.
    Then  $\psi(s)\xright{\ralpha}_N \psi(s)+\eff_N(\ralpha)(=M')$.
    Hence, by Lemma~\ref{lemma-77d}(1),
    $s\xright{\ralpha}_{\STS^\rev}\psi^{-1}(M')(=s')$.
    Thus, by the definition of $\STS^\rev$,
    $s'\xright{\alpha}_\STS s=s'\oplus \alpha$.
    We then observe that, by  Eq.(\ref{eq-ppp1}):
\[
    M'=\psi(s)+\eff_N(\ralpha)\geq
    \PRE_N(\ralpha+\beta)+\eff_N(\ralpha)=
    \PRE_N(\alpha+\beta)\;.
\]
    Hence $M'\xright{\alpha+\beta}_N$ and so, by Lemma~\ref{lemma-77d}(1),
    $s'\xright{\alpha+\beta}_{\STS^\rev}$
    and, as a consequence, $s'\xright{\alpha+\beta}_\STS$ and $s'\xright{\beta}_\STS$.
    Hence, by the definition of $\STS^\mixrev$,
    $s'\oplus\alpha\xright{\ralpha+\beta}_{\STS^\mixrev}s'\oplus\beta$.
    Moreover,
\[
\begin{array}{l@{~}c@{~}l@{~}c@{~}l@{~}c@{~}l@{~}c@{~}l@{~}c@{~}l}
    \psi(s'\oplus\alpha)
    &=&
    \psi(s')+\eff_N(\alpha)
    &=&
    M'+\eff_N(\alpha)
    &=&
    \psi(s)+\eff_N(\ralpha)+\eff_N(\alpha)
    &=&
    \psi(s)
\\
    \psi(s'\oplus\beta)
    &=&
    \psi(s')+\eff_N(\beta)
    &=&
    M'+\eff_N(\beta)
    &=&
    \psi(s)+\eff_N(\ralpha)+\eff_N(\beta)
    &=&
    M\;,
\end{array}
\]
    by Lemma~\ref{lemma-77d}(3) and Eq.(\ref{eq-ppp1}).
\end{proof}

As the next example shows, reversing a solution of $\STS$ may not
lead to a solution of $\rSTS$.
Hence, in general, one needs to
consider finding solutions to both $\STS$ and $\rSTS$.

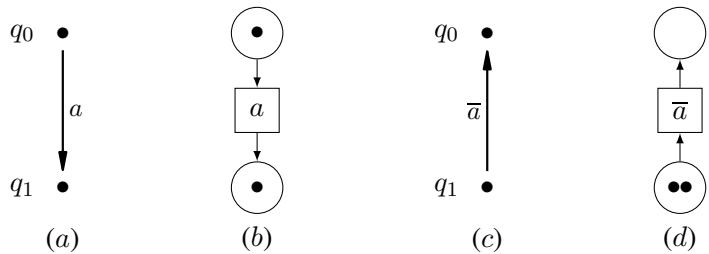
\begin{figure}[!b]
\vspace*{-1mm}
\begin{center}
\scalebox{0.97}{
\begin{tikzpicture}[node distance=1.3cm,>=arrow30,%
     line width=0.3mm,scale=0.7,bend angle=45]
     \tikzstyle{box}=[draw,regular polygon,thick,%
     regular polygon sides=4,minimum size=22mm, inner sep = -3pt]
\node (l1) at (0,-1) {($a$)};
\node (q0) [label=left:$q_0$] at (0,3) {$\bullet$};
\node (q1) [label=left:$q_1$] at (0,0) {$\bullet$};
\draw[-arrow30] (q0) to node [auto,inner sep=2pt] {\small $a$} (q1);
\end{tikzpicture}
~~~~~~~~~~~~~~~~~~
\begin{tikzpicture}[scale=0.7]
\node (l3) at (0,-1) {($b$)};
\node[circle,draw,minimum size=0.7cm] (p1) at (0,3) {};
\node (lp) at (0,3) {$\bullet$};
\node[circle,draw,minimum size=0.7cm] (p2) at (0,0) {};
\node (lx) at (0,0) {$\bullet$};

\node[draw,minimum size=0.6cm] (a) at (0,1.5){$a$};
\draw[-latex] (p1) to node [auto,inner sep=1pt] {} (a);
\draw[-latex] (a)  to node [auto,inner sep=1pt] {} (p2);
\end{tikzpicture}
~~~~~~~~~~~~~~~~~~
\begin{tikzpicture}[node distance=1.3cm,>=arrow30,%
     line width=0.3mm,scale=0.7,bend angle=45]
     \tikzstyle{box}=[draw,regular polygon,thick,%
     regular polygon sides=4,minimum size=22mm, inner sep = -3pt]
\node (l1) at (0,-1) {($c$)};
\node (q0) [label=left:$q_0$] at (0,3) {$\bullet$};
\node (q1) [label=left:$q_1$] at (0,0) {$\bullet$};
\draw[-arrow30] (q1) to node [auto,inner sep=2pt] {\small $\reverse{a}$} (q0);
\end{tikzpicture}
~~~~~~~~~~~~~~~~~~
\begin{tikzpicture}[scale=0.7]
\node (l3) at (0,-1) {($d$)};
\node[circle,draw,minimum size=0.7cm] (p1) at (0,3) { };
\node[circle,draw,minimum size=0.7cm] (p2) at (0,0) { };
\node (lx) at (0.12,0) {$\bullet$};
\node (ly) at (-0.12,0) {$\bullet$};

\node[draw,minimum size=0.6cm] (a) at (0,1.5){$\reverse{a}$};
\draw[-latex] (p2) to node [auto,inner sep=1pt] {} (a);
\draw[-latex] (a)  to node [auto,inner sep=1pt] {} (p1);
\end{tikzpicture} } \vspace*{-6mm}
\end{center}
\caption{
    Reversing a solution does not give a solution to reversing (Example~\ref{ex-ood}).
\label{fig:ttyu}}
\end{figure}

\begin{example}
\label{ex-ood}
    Let us consider $\STS$,
    a step transition system
    depicted in Figure~\ref{fig:ttyu}($a$), and its only home state
    $q_1$.
    The \textsc{pt}-net $N$ depicted Figure~\ref{fig:ttyu}($b$)
    solves $\STS$. However, the direct reverse of $N$ with the
    initial marking corresponding to $q_1$, depicted
    in Figure~\ref{fig:ttyu}($d$),  does not solve
    the step transition system $\rSTS_{q_1}$
    shown in Figure~\ref{fig:ttyu}($c$).
\hfill$\diamondsuit$
\end{example}

As the set of all the states of a step transition system is a home set,
Theorem~\ref{th-symsyn} is \emph{fundamental} as it provides
a way of solving mixed reversibility using (much) simpler synthesis problems.
In particular, if one is interested whether
the mixed reverse $\CRG_N^\mixrev$ of the
concurrent reachability graph of a \textsc{pt}-net $N$
is solvable when $\CRG_N$
has a home state.

\begin{theorem}
\label{cor-3a}
    If $r$ is a home state of $\STS$,
    then $\STS^\mixrev$ is solvable if and only if
    both $\STS$ and $\rSTS_r$ are solvable.
\end{theorem}
\begin{proof}
    Follows directly from Theorems~\ref{th-symsyn}.
\end{proof}

The above result and the proof of Theorem~\ref{th-symsyn} provide a
method for \emph{constructing} a \textsc{pt}-net implementing mixed
step reversibility provided that one can synthesise
\textsc{pt}-nets for two step transition systems using, e.g.,
theory of
regions~\cite{DBLP:conf/ac/BadouelD96,DBLP:journals/fuin/DarondeauKPY09}.

\medskip
The method for checking the
solvability of mixed reversibility easily extends
to checking  direct reversibility of set transition systems.

\begin{theorem}
\label{th-3}
    Let $\STS$ be a set transition system and
    $r$ be a home state of $\STS$.
    Then $\STS^\rev$ is solvable if and only if
    both $\STS$ and $\rSTS_r$ are solvable.
\end{theorem}

\begin{proof}
($\Longrightarrow$)
    Let $\STS^\rev \simeq_\psi \CRG_N$ .
    Then $\STS\simeq_\psi\CRG_{N|_T}$ and
    $\rSTS_r\simeq_\psi \CRG_{N'}$, where $N'$ is
    $N|_{\rT}$ with the initial marking set to $\psi(r)$.\medskip

\noindent ($\Longleftarrow$)
    Follows from Theorems~\ref{th-symsyn}
    and~\ref{cor-ddd}.
\end{proof}

\section{From sequential reversibility to step reversibility}
\label{sect-3}

Checking the feasibility of step reversibility
is, in general, a difficult task.
The next result shows that
in certain cases it is possible to proceed
more effectively,
if one is given a \textsc{pt}-net that solves the original
step transition system, over-approximates its
reverse containing only spikes,
and under-approximates its mixed reverse.

\begin{theorem}
\label{th-2ba}
    Let $\STS=(S,T,\to,s_0)$ be a \textsc{cest}-system
    and $N=(P,T\cup\rT,F,M_0)$ be a \textsc{pt}-net such that:
\begin{equation}
\label{eq-111}
    (\STS^\spike)^\rev\lhd\CRG_N\lhd\STS^\mixrev
    ~~~~~~~~\textrm{and}~~~~~~~~
    \STS\simeq\CRG_{N|_T} \;.
\end{equation}
    Then $\STS^\mixrev$ is solvable.
\end{theorem}

\begin{proof}
    The states as well as the initial states
    of $(\STS^\spike)^\rev$, $\STS^\mixrev$, and $\STS$ are
    the same.
    Moreover, $((\STS^\spike)^\rev|_T)^\seq=(\STS^\mixrev|_T)^\seq=\STS^\seq$.
    Similarly, the initial states
    of $\CRG_N$ and $\CRG_{N|_T}$ are the same and we have
    $(\CRG_N)|_T=\CRG_{N|_T}$.
    We also observe that
    all step transition systems in Eq.(\ref{eq-111})
    are \textsc{cest}-systems,
    and there is a unique bijection $\psi$ such that:
\begin{equation}
\label{eq-7777}
    (\STS^\spike)^\rev\lhd_\psi\CRG_N
    ~~~~~~~~~~~~~
    \CRG_N
    \lhd_{\psi^{-1}}
    \STS^\mixrev
    ~~~~~~~~~~~~~
    \STS\simeq_\psi\CRG_{N|_T} \;.
\end{equation}
    By the first part of Eq.(\ref{eq-111}), \textit{SEQ}, and the fact that
    we may assume that
    each action in $ T$ appears in the labels of the transitions of $\STS$,
    we have:
\begin{equation}
\label{eq-7777a}
    \reach_N=\reach_{N|_T}
    ~~~~~~~~\mbox{and}~~~~~~~
    \eff_N(a)=-\eff_N(\reverse{a})~~~\mbox{for every $a\in T$}\;.
\end{equation}

\begin{lemma}
\label{eq-113}
    It can be assumed that
    $\PRE_N(\reverse{a})\geq \POST_N(a)$
    and
    $\POST_N(\reverse{a})\geq \PRE_N(a)$,
    for every $a\in T$.
\end{lemma}
\begin{proof}[Lemma~\ref{eq-113}]
    Suppose that
    $F(p,\reverse{a})< F(a,p)$, and so also $F(\reverse{a},p) > F(p,a)$.
    We then modify $F$ to become $F'$ which is the same as $F$
    except that $F'(p,\reverse{a})= F(a,p)$ and
    $F'(\reverse{a},p)= F(p,a)$.
    Let $N'$ be the resulting \textsc{pt}-net. Clearly,
    $\eff_N=\eff_{N'}$.

\medskip
    After this modification, which does not affect
    actions in $T$, the second part of
    Eq.(\ref{eq-111}) is still satisfied after taking $N'$
    to play the role of $N$.
    However, the first part of
    Eq.(\ref{eq-111}) needs to be demonstrated.

    We observe that the modification
    can only restrict
    the enabling of steps involving $\reverse{a}$.
    Hence, if the first part of Eq.(\ref{eq-111}) does not hold with $N'$
    playing the role of $N$, then
    there is $M\in\reach_{N'}\subseteq\reach_N$ and $k\geq 1$
    such that
    $M\xright{\reverse{a}^k}_{N}M'$ \textit{(*)}
    and $\neg M\xright{\reverse{a}^k}_{N'}$ \textit{(**)}.
    By Eq.(\ref{eq-7777a})
    and \textit{(*)}, we have
    $M'\xright{a^k}_{N}M$, and so
    $M(p)\geq \POST_N(a^k)(p)$ \textit{(***)}.

    By construction,
    \textit{(**)} implies $\PRE_{N'}(\reverse{a}^k)(p)>M(p)$.
    Thus, by
    $\PRE_{N'}(\reverse{a}^k)(p)=\POST_N(a^k)(p)$,
    we obtain $\POST_N(a^k)(p)>M(p)$,
    yielding a contradiction with \textit{(***)}.

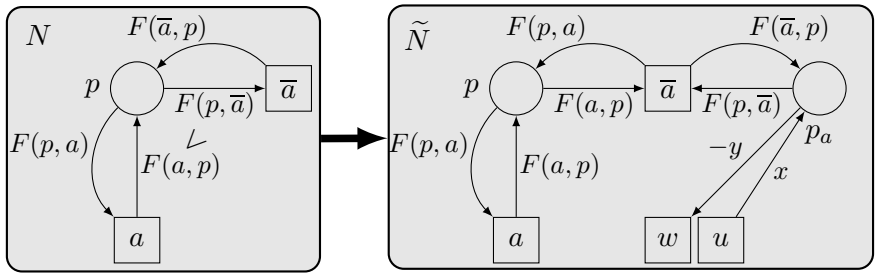
\begin{figure}[!b]
\vspace{1mm}
\begin{center}
\begin{tikzpicture}[scale=1.0]
\node[] (n1) at(-1.3,2.7){$N$};
\node[circle,draw,minimum size=0.7cm] (p) at (0,2) [label=left:$p$] {};
\node[rotate=45] (x1) at(.75,1.3){\large $<$};
\node[draw,minimum size=0.6cm] (t) at (0,0){$a$};
\node[draw,minimum size=0.6cm] (to) at (2,2){$\reverse{a}$};
\node[] (n2) at(3.7,2.7){$\widetilde{N}$};
\node[circle,draw,minimum size=0.7cm] (p2) at (5,2) [label=left:$p$] {};
\node[circle,draw,minimum size=0.7cm] (pt) at (9,2) [label=below:$p_a$] {};
\node (x2) at(9.5,2){};
\node[draw,minimum size=0.6cm] (t2) at (5,0){$a$};
\node[draw,minimum size=0.6cm] (to2) at (7,2){$\reverse{a}$};
\node[draw,minimum size=0.6cm] (tt) at (7.7,0){$u$};
\node[draw,minimum size=0.6cm] (tw) at (7,0){$w$};
%
\draw[-latex] (t) to node [auto,swap,inner sep=1pt]
{\small $F(a,p)$} (p);
\draw[-latex] (p) to [out=-135,in=135] node [auto,swap,inner sep=1pt]
{\small $F(p,a)$} (t);
\draw[-latex] (p) to node [auto,swap,inner sep=1pt]
{\small $F(p,\reverse{a})$} (to);
\draw[-latex] (to) to [out=135,in=45] node [auto,swap,inner sep=1pt]
{\small $F(\reverse{a},p)$} (p);
%
\draw[-latex] (t2) to node [auto,swap,inner sep=1pt] {\small $F(a,p)$} (p2);
\draw[-latex] (p2) to [out=-135,in=135] node [auto,swap,inner sep=1pt]
{\small $F(p,a)$} (t2);
\draw[-latex] (p2) to node [auto,swap,inner sep=1pt] {\small $F(a,p)$} (to2);
\draw[-latex] (to2) to [out=135,in=45] node [auto,swap,inner sep=1pt]
{\small $F(p,a)$} (p2);
\draw[-latex] (pt) to node [auto,inner sep=1pt] {\small $F(p,\reverse{a})$} (to2);
\draw[-latex] (to2) to [out=45,in=135] node [auto,inner sep=1pt]
{\small $F(\reverse{a},p)$} (pt);

\draw[-latex] (tt) to node [auto,swap,inner sep=1pt] {\small $x$} (pt);
\draw[-latex] (pt) to node [auto,swap,inner sep=1pt] {\small $-y$} (tw);

\begin{pgfonlayer}{background}
\node[fill=gray!20,thick,draw=black,rounded corners=2mm,inner sep=3]
(bg1) [fit = (n1) (t) (to)] {};
\end{pgfonlayer}

\begin{pgfonlayer}{background}
\node[fill=gray!20,thick,draw=black,rounded corners=2mm,inner sep=2]
(bg2) [fit = (n2) (tt) (x2)] {};
\end{pgfonlayer}

\draw[-latex,line width=1mm] (bg1) -- (bg2);
\end{tikzpicture}\vspace*{-3mm}
\end{center}
\caption{Introducing place $p_a$
 in the proof of Theorem~\ref{th-2ba}, where $u$ represents
 any place in $T\cup\rT\setminus\{\reverse{a}\}$ for which $x=\eff_N(u)(p)>0$,
 and $w$ any place for which $y=\eff_N(w)(p)\leq 0$.
\label{fig-1}}
\end{figure}

    We can apply the above modification as many times
    as needed,
    finally concluding that the result holds,
    as any modification does not invalidate the conditions
    captured in the formulation of this lemma
    that were obtained by the previous modifications.
\end{proof}

    We will show that $\STS^\mixrev$ is solvable by a \textsc{pt}-net
    $\wN=(\wP,T\cup\rT,\wF,\wM_0)$
    constructed thus:
\begin{itemize}
\item
    $\wP=\bigcup_{p\in P}P_p$, where, for every $p\in P$,\footnote
    {
    Intuitively, each $p_a\in P_p$ is
    a (suitably adjusted) copy of $p$.
    }
    $P_p=\{p\}\cup\{p_a\mid a\in T\wedge F(p,\reverse{a})>F(a,p)\}$
    and
    $\wM_0(q)=M_0(p)$ for $q\in P_p$.

\item
    The connections in $\wN$ are set
    as follows, where $p\in P$ and $u\in T\cup\rT\setminus\{\reverse{a}\}$:
\begin{itemize}
\item
    $\wF(p,\reverse{a})=F(a,p)$
    and
    $\wF(\reverse{a},p)=F(p,a)$.
\item
    $\wF(p_a,\reverse{a})=F(p,\reverse{a})$
    and
    $\wF(\reverse{a},p_a)=F(\reverse{a},p)$.
\item
    $\eff_N(u)(p)>0$ implies
    $\wF(p_a,u)=0$
    and
    $\wF(u,p_a)=\eff_N(u)(p)$.
\item
    $\eff_N(u)(p)\leq 0$ implies
    $\wF(u,p_a)=0$
    and
    $\wF(p_a,u)=-\eff_N(u)(p)$.
\item
    $\wF$ on $(P\times T)\cup(T\times P)$
    is as $F$ unless it has been set explicitly above.
\end{itemize}
\end{itemize}
    In what follows, for every marking $M$ of $N$,
    we use $\phi(M)$ to denote the marking of $\wN$
    such that
    $\phi(M)(q)=M(p)$, for every $q\in P_p$ and $p\in P$.
    Hence $\phi(M_0)=\wM_0$.

\medskip
    We now present a number of straightforward
    properties of $\wN$.
    We first observe that,
    by Lemma~\ref{eq-113},
    for all $a\in T$, $u\in T\cup\rT$, and $p\in P$,
\begin{equation}
\label{eq-113a}
\begin{array}{l@{~}c@{~}l@{~~~~~~~~~~}r@{~}c@{~}l}
    \PRE_\wN(\reverse{a})
    &\geq&
    \POST_\wN(a)
    &
    \eff_\wN(a)
    &=&
    -\eff_\wN(\reverse{a})
\\
    \POST_\wN(\reverse{a})
    &\geq&
    \PRE_\wN(a)
    &
    \eff_\wN(u)(P_p)
    &=&
    \{\eff_N(u)(p)\}\;.
\end{array}
\end{equation}
    Therefore, for every marking $M$ of $N$
    and every $\kappa\in\mult(T\cup\rT)$ such that
    $M+\eff_N(\kappa)\geq\es$,
\begin{equation}
\label{eq-4000}
    \phi(M)+\eff_\wN(\kappa)=\phi(M+\eff_N(\kappa))\;.
\end{equation}
    The
    construction does not affect the enabling of
    steps involving just one action as well as
    steps $\alpha$
    over $T$ since $p_a\in P_p$ cannot disable
    $\alpha$ if it is not also disabled by $p$.
    Hence, for all markings
    $M$ of $N$, $u\in T\cup\rT$, $k\geq 1$,
    and $\alpha\in\mult(T)$:
\begin{equation}
\label{eq-4001}
    M\xright{u^k}_N
    \iff
    \phi(M)\xright{u^k}_\wN
    ~~~~~~~~\mbox{and}~~~~~~~~
    M\xright{\alpha}_N
    \iff
    \phi(M)\xright{\alpha}_\wN\;.
\end{equation}
    Thus, by Eqs.(\ref{eq-7777},\ref{eq-4000},\ref{eq-4001})
    and $\wM_0=\phi(M_0)$,
\begin{equation}
\label{eq-119}
    (\STS^\spike)^\rev \lhd_{\phi\circ \psi} \CRG_\wN
    ~~~~~~\mbox{and}~~~~~~
    \STS
    \simeq_{\phi\circ \psi}
    \CRG_{\wN|_T }
    \simeq_{\phi^{-1}}
    \CRG_{N|_T }\;.
\end{equation}

\begin{lemma}
\label{lem-999}
    Let $\alpha,\beta\in\mult(T)$
    and $\wM=\phi(M)$, for some $M\in\mult(P)$.
\begin{enumerate}
\item
    $\wM\xright{\ralpha+\beta}_\wN$
    implies
    $
    \wM-\eff_\wN(\alpha)
    \xright{\alpha+\beta}_\wN
    \wM+\eff_\wN(\beta)$.

\item
    $\wM\xright{\alpha+\beta}_\wN$
    implies
    $
    \wM+\eff_\wN(\alpha)
    \xright{\ralpha+\beta}_\wN
    \wM+\eff_\wN(\beta)$.
\end{enumerate}
\end{lemma}

\begin{proof}[Lemma~\ref{lem-999}]
(1)
    We first observe that, by \textit{SEQ},
    $\wM-\eff_\wN(\alpha)=\wM+\eff_\wN(\ralpha) \in\reach_\wN$.
    We then observe that, by $\wM\geq \PRE_\wN(\ralpha+\beta)$,
    the step
    $\alpha+\beta$ is enabled at $\wM-\eff_\wN(\alpha)$, and so, by Eq.(\ref{eq-113a}):
\[
    \wM-\eff_\wN(\alpha)
    \geq
    \PRE_\wN(\ralpha+\beta) -\eff_\wN(\alpha)
    =
    \PRE_\wN(\ralpha)+\PRE_\wN(\beta)-\POST_\wN(\alpha)+\PRE_\wN(\alpha)
    \geq
    \PRE_\wN(\alpha+\beta).
\]
    Hence, the result holds, as
    $\wM-\eff_\wN(\alpha)+\eff_\wN(\alpha+\beta)=\wM+\eff_\wN(\beta)$.

(2)
    By \textit{SEQ},
    $\wM\xright{\alpha}_\wN \wM+\eff_\wN(\alpha)(=M')$.
    Suppose that
    $M'\xright{\ralpha+\beta}_\wN$ does not hold.
    Then there is $q\in \wP$ such that
    $\PRE_\wN(\ralpha+\beta)(q)>M'(q)$ \textit{(*)}.
    Moreover,
    $\wM\geq \PRE_\wN(\alpha+\beta)$.
    Hence, we have:
\[
    \PRE_\wN(\ralpha+\beta)(q)
    >
    \wM(q)+\eff_\wN(\alpha)(q)
    \geq
    \PRE_\wN(\alpha+\beta)(q)
    +\eff_\wN(\alpha)(q)
    =
    \PRE_\wN(\beta)(q)
    +\POST_\wN(\alpha)(q),
\]
    and so
    $\PRE_\wN(\ralpha)(q)>\POST_\wN(\alpha)(q)$.
    Thus there is $a\in\alpha$ and such that
    $\wF(q,\reverse{a})>\wF(a,q)$ and so,
    by the definition of $\wN$,
    $q=p_a$, for some $p\in P$.
    Now, it follows from the construction of $\wN$,
    that there are $\alpha_0,\alpha_1,\beta_0,\beta_1$
    and $k\geq 1$
    such that $\alpha=a^k+\alpha_0+\alpha_1$ and
    $\beta=\beta_0+\beta_1$ and $a\not\in \alpha_0+\alpha_1$
    and, for $x=\alpha,\beta$,
    we have:
\[
\begin{array}{l@{~}l@{~}l@{~}l@{~}l@{~}l@{~}l@{~}l@{~}l}
    \POST_\wN(x_1)(p_a)
    &=&
    \PRE_\wN(x_0)(p_a)
    &=&
    0
    &=&
    \PRE_\wN(\reverse{x_1})(p_a)
    &=&
    \POST_\wN(\reverse{x_0})(p_a)
\\
    \PRE_\wN(\reverse{x_0})(p_a)
    &=&
    \POST_\wN(x_0)(p_a)
    & & &&
    \PRE_\wN(x_1)(p_a)
    &=&
    \POST_\wN(\reverse{x_1})(p_a) \;.
\end{array}
\]
    By \textit{SEQ},
    $
    \wM
    \xright{\alpha_1+\beta_1}_\wN \wM+\eff_\wN(\alpha_1+\beta_1)
    \xright{a^k}_\wN \wM+\eff_\wN(\alpha_1+\beta_1+a^k)$.
    Thus, by Eq.(\ref{eq-119}),
    $
    \wM+\eff_\wN(\alpha_1+\beta_1+a^k)
    \xright{\reverse{a}^k}_\wN \wM+\eff_\wN(\alpha_1+\beta_1)
    $,
    and so we have:
\[
\begin{array}{lcl}
    \wM(p_a)+\eff_\wN(\alpha_1+\beta_1+a^k)(p_a)
    &=&
    \wM(p_a)+\eff_\wN(a^k)(p_a)+
    \eff_\wN(\alpha_1+\beta_1)(p_a)
\\
    &=&
    \wM(p_a)+\eff_\wN(a^k)(p_a)-
    \PRE_\wN(\alpha_1+\beta_1)(p_a)
\\
    &\geq&
    \PRE_\wN(\reverse{a}^k)(p_a)  \;.
\end{array}
\]
    We therefore have:
\[
\begin{array}{lcl}
    M'(p_a)
    &=&
    M(p_a)+\eff_\wN(a^k)(p_a)-\PRE_\wN(\alpha_1)(p_a)+\POST_\wN(\alpha_0)(p_a)
\\
    &\geq&
    \PRE_\wN(\reverse{a}^k)(p_a)+\PRE_\wN(\beta_1)(p_a)+\POST_\wN(\alpha_0)(p_a)
\\
    &=&
    \PRE_\wN(\reverse{a}^k)(p_a)+\PRE_\wN(\beta_1)(p_a)+
    \PRE_\wN(\ralpha_0)(p_a)
\\
    &=&
    \PRE_\wN(\ralpha)(p_a)+\PRE_\wN(\beta)(p_a)
\\
    &=&
    \PRE_\wN(\ralpha+\beta)(p_a)\;,
\end{array}
\]
    yielding
    a contradiction with \textit{(*)}.
    Thus $M'\xright{\ralpha+\beta}_\wN$ holds.
    Hence we obtain the result as
    we have
    $M'+\eff_\wN(\ralpha+\beta)=\wM+\eff_\wN(\alpha)
    +\eff_\wN(\ralpha+\beta)=
    \wM+\eff_\wN(\beta)$.
\end{proof}

    We now conclude that
    $\STS^\mixrev\simeq_{\phi\circ\psi}\CRG_\wN$ holds
    thanks to Eq.(\ref{eq-119}) and Lemma~\ref{lem-999}.
\end{proof}

The last result leads to a simple sufficient
condition for
the solvability of direct reversibility in the case
that proper multisets are not involved.

\begin{theorem}
\label{corollary-2ba}
    Let $\STS$ be a solvable set \textsc{cest}-system
    such that $(\STS^\seq)^\rev$ is solvable.
    Then $\STS^\rev$ is solvable.
\end{theorem}
\begin{proof}
    Referring to the notation and proof of Theorem~\ref{th-2ba},
    we construct a new net $\wN'$, by adding to $\wN$ a fresh set of
    (mutex) places
    $P'=\{p_{ab}\mid a,b\in T\}$, where each $p_{ab}$
    is such that $\wM_0(p_{ab})=1$ and has four non-zero connections:
    $
    \wF(a,p_{ab})=\wF(p_{ab},a)=
    \wF(\reverse{b},p_{ab})=\wF(p_{ab},\reverse{b})=1$.

    Since all the steps in $\STS$ are sets
    $P'$ ensure that each step
    enabled at a reachable marking of $\wN'$
    is a subset of $T$ or a subset of $\rT$.
    Moreover, the enabling of such steps is not affected
    by adding $P'$, so we obtain
    $\STS^\rev\simeq\CRG_{\wN'}$
    as $\STS^\mixrev\simeq\CRG_{\wN}$ holds by Theorem~\ref{th-2ba}.
\end{proof}

As the next example shows, modifying the original \textsc{pt}-net
in Theorem~\ref{th-2ba} is unavoidable.

\begin{example}
    Figure~\ref{fig-4d}($a$) depicts a family $N_{n,m}$ of
    \textsc{pt}-nets which satisfy the
    assumptions of Theorem~\ref{th-2ba}.
    We  have $\CRG_{N_{n,m}}\not\simeq\STS^\mixrev$, where
    $\STS$ is the step reachability graph of the
    \textsc{pt}-net obtained from $N_{n,m}$
    after deleting actions $\reverse{a}$ and $\reverse{b}$.
    However, the construction from the
    proof of Theorem~\ref{th-2ba} yields
    the \textsc{pt}-net $\CRG_{\widetilde{N}_{n,m}}$,
    shown in Figure~\ref{fig-4d}($b$), satisfying
    $\CRG_{\widetilde{N}_{n,m}} \simeq\STS^\mixrev$.
\hfill$\diamondsuit$
\end{example}

\begin{figure}[h]
\vspace*{-2mm}
\begin{center}
\begin{tikzpicture}[scale=0.64]

\node[circle,draw,minimum size=0.7cm] (p1) at (1,3) {$m$};
\node[circle,draw,minimum size=0.7cm] (p2) at (7,3) {$n$};
\node[circle,draw,minimum size=0.7cm] (p3) at (1,0) {$ $};
\node[circle,draw,minimum size=0.7cm] (p6) at (7,0) {$ $};
\node[circle,draw,minimum size=0.7cm] (p7) at (4,1.5) {$k$};

\node[draw,minimum size=0.6cm] (a1) at (0,1.5){$a$};
\node[draw,minimum size=0.6cm] (a2) at (2,1.5){$\reverse{a}$};
\node[draw,minimum size=0.6cm] (b2) at (6,1.5){$\reverse{b}$};
\node[draw,minimum size=0.6cm] (b1) at (8,1.5){$b$};

\draw[-latex] (p1) to node [auto,inner sep=1pt] {} (a1);
\draw[-latex] (a1) to node [auto,inner sep=1pt] {} (p3);
\draw[-latex] (p3) to node [auto,inner sep=1pt] {} (a2);
\draw[-latex] (a2) to node [auto,inner sep=1pt] {} (p1);

\draw[-latex] (p2) to node [auto,inner sep=1pt] {} (b1);
\draw[-latex] (b1) to node [auto,inner sep=1pt] {} (p6);
\draw[-latex] (p6) to node [auto,inner sep=1pt] {} (b2);
\draw[-latex] (b2) to node [auto,inner sep=1pt] {} (p2);

\draw[-latex] (p7) to node [auto,inner sep=1pt] {} (a2);
\draw[-latex] (a2) to node [auto,inner sep=1pt] {} (p7);
\draw[-latex] (p7) to node [auto,inner sep=1pt] {} (b2);
\draw[-latex] (b2) to node [auto,inner sep=1pt] {} (p7);
\end{tikzpicture}
($a$)
~~~~~~~~~~~~
($b$)
\begin{tikzpicture}[scale=0.64]

\node[circle,draw,minimum size=0.7cm] (p1) at (1,3) {$m$};
\node[circle,draw,minimum size=0.7cm] (p2) at (7,3) {$n$};
\node[circle,draw,minimum size=0.7cm] (p3) at (1,0) {$ $};
\node[circle,draw,minimum size=0.7cm] (p6) at (7,0) {$ $};
\node[circle,draw,minimum size=0.7cm] (p7) at (4,1.5) {$k$};
\node[circle,draw,minimum size=0.7cm] (p7a)at (4,0) {$k$};
\node[circle,draw,minimum size=0.7cm] (p7b)at (4,3) {$k$};

\node[draw,minimum size=0.6cm] (a1) at (0,1.5){$a$};
\node[draw,minimum size=0.6cm] (a2) at (2,1.5){$\reverse{a}$};
\node[draw,minimum size=0.6cm] (b2) at (6,1.5){$\reverse{b}$};
\node[draw,minimum size=0.6cm] (b1) at (8,1.5){$b$};

\draw[-latex] (p1) to node [auto,inner sep=1pt] {} (a1);
\draw[-latex] (a1) to node [auto,inner sep=1pt] {} (p3);
\draw[-latex] (p3) to node [auto,inner sep=1pt] {} (a2);
\draw[-latex] (a2) to node [auto,inner sep=1pt] {} (p1);

\draw[-latex] (p2) to node [auto,inner sep=1pt] {} (b1);
\draw[-latex] (b1) to node [auto,inner sep=1pt] {} (p6);
\draw[-latex] (p6) to node [auto,inner sep=1pt] {} (b2);
\draw[-latex] (b2) to node [auto,inner sep=1pt] {} (p2);

\draw[-latex] (p7a) to node [auto,inner sep=1pt]{} (a2);
\draw[-latex] (a2) to node [auto,inner sep=1pt] {} (p7a);
\draw[-latex] (p7b) to node [auto,inner sep=1pt]{} (b2);
\draw[-latex] (b2) to node [auto,inner sep=1pt] {} (p7b);
\end{tikzpicture}\vspace*{-4mm}
\end{center}
\caption{
    \textsc{pt}-net $N_{n,m}$ with
    $k=\max(m,n)$ and $m,n\geq1$ ($a$);
    and the same net after applying the construction from
    Theorem~\ref{th-2ba} ($b$).
\label{fig-4d}}
\end{figure}
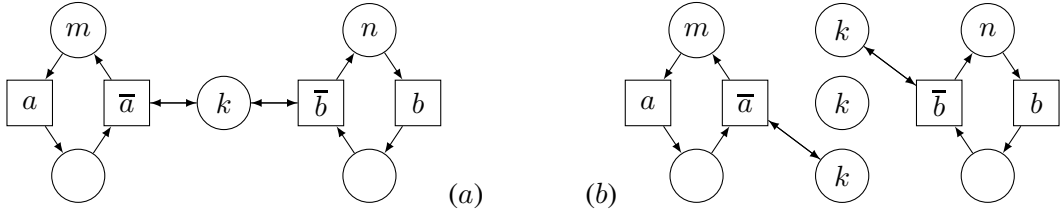

It is not possible to
drop Eq.\eqref{eq-111} from the formulation of
Theorem~\ref{th-2ba}.
The next example shows a \textsc{cest}-system which has
only one non-singleton step and is reversible in the sequential semantics,
but cannot be reversed in step
sequence semantics, even with mixed reverses.

\begin{example}
\label{ex:spikecex}
    Let us consider a step transition system $\STS$  together with a
    \textsc{pt}-net solving
    it, shown in Figure~\ref{fig:aabb}($a,b$).
    If we erase the spike between the states $v_0$ and $v_2$, and
    add all the reverses (see Figure~\ref{fig:aabb}($c$)),
    then the resulting step transition system is solvable
    (see Figure~\ref{fig:aabb}($d$)).
    However, $\STS$ cannot be reversed, as shown below.

\begin{figure}[!h]
\begin{center}
($a$)
\begin{tikzpicture}[node distance=1.3cm,>=arrow30,%
     line width=0.3mm,scale=1.0,bend angle=45]
     \tikzstyle{box}=[draw,regular polygon,thick,%
     regular polygon sides=4,minimum size=22mm, inner sep = -3pt]
\node (v0) [label=left:$v_0$] at (0,3) {$\bullet$};
\node (v1) [label=left:$v_1$] at (0,1.5) {$\bullet$};
\node (v2) [label=below left:$v_2$] at (0,0) {$\bullet$};
\node (v3) [label=below:$v_3$] at (1.5,0) {$\bullet$};
\node (v4) [label=below:$v_4$] at (3,0) {$\bullet$};

\draw[-arrow30] (v0) to node [auto,swap,inner sep=2pt] {\small $a$} (v1);
\draw[-arrow30] (v1) to node [auto,swap,inner sep=2pt] {\small $a$} (v2);
\draw[-arrow30] (v2) to node [auto,swap,inner sep=2pt] {\small $b$} (v3);
\draw[-arrow30] (v3) to node [auto,swap,inner sep=2pt] {\small $b$} (v4);
\draw[-arrow30] (v0) to [out=-60,in=60] node [auto,inner sep=1pt] {\small $aa$} (v2);
\end{tikzpicture}
~~~~~~~~~~~~~~~
\begin{tikzpicture}[scale=0.92]

\node[circle,draw,minimum size=0.7cm] (p1) at (0,3) {$\bullet\bullet$};
\node[circle,draw,minimum size=0.7cm] (p2) at (0,0) { };
\node[circle,draw,minimum size=0.7cm] (p3) at (3,0) {$\bullet\bullet$};

\node[draw,minimum size=0.6cm] (a1) at (0,1.5){$a$};
\node[draw,minimum size=0.6cm] (b1) at (1.5,0){$b$};

\draw[-latex] (p1) to node [auto,inner sep=1pt] {} (a1);
\draw[-latex] (a1) to node [auto,inner sep=1pt] {} (p2);
\draw[-latex] (p2) to [out=30, in=150] node [auto,inner sep=2pt] {\textit{2}} (b1);
\draw[-latex] (b1) to [out=-150, in=-30] node [auto,inner sep=2pt] {\textit{2}} (p2);
\draw[-latex] (p3) to node [auto,inner sep=1pt] {} (b1);
\end{tikzpicture}
($b$)\\[0.8cm]

($c$)
\begin{tikzpicture}[node distance=1.3cm,>=arrow30,%
     line width=0.3mm,scale=1.0,bend angle=45]
     \tikzstyle{box}=[draw,regular polygon,thick,%
     regular polygon sides=4,minimum size=22mm, inner sep = -3pt]
\node (v0) [label=left:$v_0$] at (0,3) {$\bullet$};
\node (v1) [label=left:$v_1$] at (0,1.5) {$\bullet$};
\node (v2) [label=below left:$v_2$] at (0,0) {$\bullet$};
\node (v3) [label=below:$v_3$] at (1.5,0) {$\bullet$};
\node (v4) [label=below:$v_4$] at (3,0) {$\bullet$};

\draw[-arrow30] (v0) to node [auto,swap,inner sep=2pt] {\small $a$} (v1);
\draw[-arrow30] (v1) to node [auto,swap,inner sep=2pt] {\small $a$} (v2);
\draw[-arrow30] (v2) to node [auto,swap,inner sep=2pt] {\small $b$} (v3);
\draw[-arrow30] (v3) to node [auto,swap,inner sep=2pt] {\small $b$} (v4);
\draw[-arrow30] (v2) to [out=60,in=-60] node [auto,swap,inner sep=1pt]
{\small $\reverse{a}$} (v1);
\draw[-arrow30] (v1) to [out=60,in=-60] node [auto,swap,inner sep=1pt]
{\small $\reverse{a}$} (v0);
\draw[-arrow30] (v3) to [out=150,in=30] node [auto,swap,inner sep=1pt]
{\small $\reverse{b}$} (v2);
\draw[-arrow30] (v4) to [out=150,in=30] node [auto,swap,inner sep=1pt]
{\small $\reverse{b}$} (v3);
\end{tikzpicture}
~~~~~~~~~~~~~~~
\begin{tikzpicture}[scale=0.92]

\node[circle,draw,minimum size=0.7cm] (p1) at (0,3) { };
\node (m1) at (0,2.9) {$\bullet\bullet$};
\node (m2) at (0,3.1) {$\bullet$};
\node[circle,draw,minimum size=0.7cm] (p2) at (0,0) { };
\node[circle,draw,minimum size=0.7cm] (p3) at (3,0) { };
\node (m3) at (3,0) {$\bullet\bullet$};
\node[circle,draw,minimum size=0.7cm] (p4) at (1.5,-1.5) { };

\node[draw,minimum size=0.6cm] (a1) at (0,1.5){$a$};
\node[draw,minimum size=0.6cm] (a2) at (1.5,1.5){$\reverse{a}$};
\node[draw,minimum size=0.6cm] (b2) at (0,-1.5){$\reverse{b}$};
\node[draw,minimum size=0.6cm] (b1) at (1.5,0){$b$};

\draw[-latex] (p1) to [out=-120, in=120] node [auto,swap,inner sep=1pt] {\textit{2}} (a1);
\draw[-latex] (a1) to [out=60, in=-60] node [auto,inner sep=1pt] {} (p1);
\draw[-latex] (a1) to node [auto,inner sep=1pt] {} (p2);
\draw[-latex] (p2) to node [auto,inner sep=1pt] {} (a2);
\draw[-latex] (a2) to node [auto,inner sep=1pt] {\textit{2}} (p1);
\draw[-latex] (p1) to [out=-10, in=100] node [auto,inner sep=1pt] {} (a2);
\draw[-latex] (a2) to [out=-20, in=110] node [auto,inner sep=1pt] {\textit{2}} (p3);
\draw[-latex] (p3) to [out=160, in=-70] node [auto,inner sep=1pt] {\textit{2}} (a2);
\draw[-latex] (p2) to [out=30, in=150] node [auto,inner sep=1pt] {\textit{2}} (b1);
\draw[-latex] (b1) to [out=-150, in=-30] node [auto,inner sep=1pt] {\textit{2}} (p2);
\draw[-latex] (p3) to node [auto,inner sep=1pt] {} (b1);
\draw[-latex] (b1) to node [auto,inner sep=1pt] {} (p4);
\draw[-latex] (b2) to [out=60, in=-60] node [auto,swap,inner sep=1pt] {\textit{2}} (p2);
\draw[-latex] (p2) to [out=-120, in=120] node [auto,swap,inner sep=1pt] {\textit{2}} (b2);
\draw[-latex] (p4) to node [auto,inner sep=1pt] {} (b2);
\draw[-latex] (b2) to [out=-40, in=-90] node [auto,inner sep=1pt] {} (p3);c
\end{tikzpicture}
($d$)
\end{center}\vspace*{-2mm}
\caption{A step transition system   $\STS$ with one spike ($a$), and a    \textsc{pt}-net
    solving it ($b$).     $\STS$ without the spike between $v_0$ and $v_2$ can    be reversed ($c,d$), but $\STS$ cannot.
\label{fig:aabb}}\vspace*{-2mm}
\end{figure}
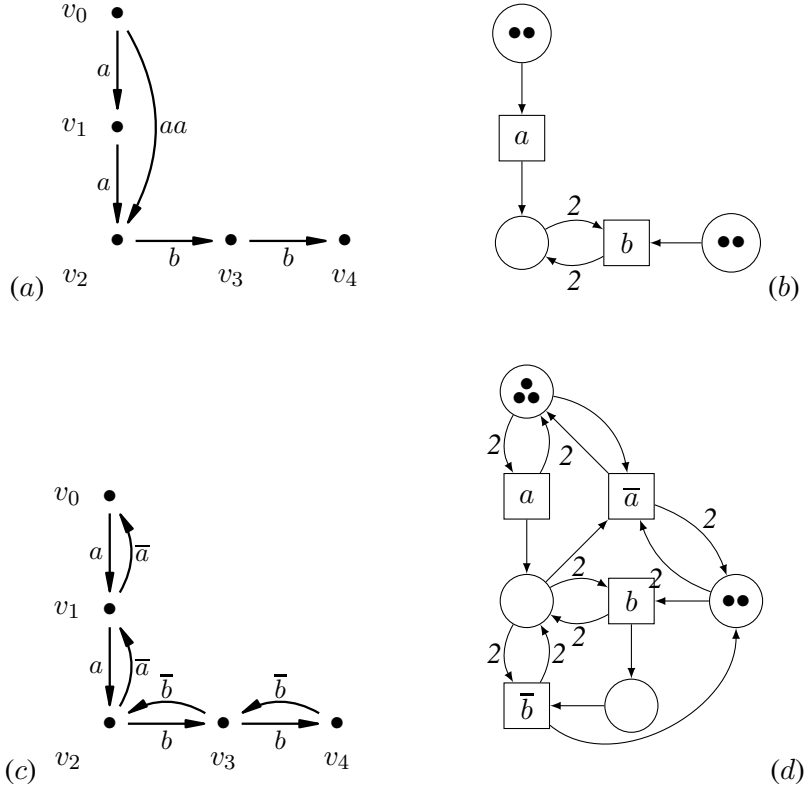

    Suppose that there is a \textsc{pt}-net $N$
    solving $\STS^\mixrev$.
    Let $M_i$ be the marking of $N$ corresponding to the state $v_i$,
    for $i=0,\dots,4$.
    Then the step $(\reverse{aa})$ is enabled at $M_2$, and
    $\reverse{a}$ is not be enabled at $M_3$ \textit{(*)}.

\medskip
    Let $p$ be any place of $N$.
    We first observe that $M_4$ is a marking, and so
    $0\leq M_4(p)=M_2(p)+2k$, where $k=\eff_N(b)(p)$.
    Hence $\frac{1}{2}\cdot M_2(p)+k \geq 0$.
    We then recall that $(\reverse{aa})$ is enabled at $M_2$, and so
    $M_2(p)\geq 2\cdot F(p,\reverse{a})$.
    Hence $\frac{1}{2}\cdot M_2(p)\geq F(p,\reverse{a})$.
    We therefore have:
\[
    M_3(p)=M_2(p)+k
    =
    \frac{1}{2}\cdot M_2(p)+k+\frac{1}{2}\cdot M_2(p)
    \geq
    0+F(p,\reverse{a})
    =F(p,\reverse{a})\;.
\]
    This means that $\reverse{a}$ is a step enabled at $M_3$,
    yielding a contradiction with \textit{(*)}.
\hfill$\diamondsuit$
\end{example}

One might expect that, as it was shown to
be the case for bounded \textsc{pt}-nets executed under
the sequential semantics~\cite{csp2016}, it is sufficient to
use \textsc{pt}-nets with split reverses also
for the reversing under the step semantics.
This, however, is not the case
as demonstrated in the following example.

\begin{example}
\label{ex:splitneg}
    Let us consider a step transition system $\STS$  together
    with a \textsc{pt}-net solving it, shown in Figure~\ref{fig-ex2}($a,b$).
    Suppose that there is a \textsc{pt}-net $N$ with split reverses
    such that  $\CRG_N$ is a split reverse
    of $\STS$.
    Moreover,
    let $M_i$ be the marking of $N$ corresponding to
    $v_i$, for $i=1,\dots,6$.

\begin{figure}[!h]
\vspace*{-2mm}
\begin{center}
($a$)~~
\begin{tikzpicture}[node distance=1.3cm,>=arrow30,%
     line width=0.3mm,scale=1.0,bend angle=45]
     \tikzstyle{box}=[draw,regular polygon,thick,%
     regular polygon sides=4,minimum size=22mm, inner sep = -3pt]

\node (v1) [label=$v_1$]     at (4,4) {$\bullet$};
\node (v2) [label=left:$v_2$] at (2,3) {$\bullet$};
\node (v3) [label=right:$v_3$] at (6,3) {$\bullet$};
\node (v4) [label=below:$v_4$] at (4,2) {$\bullet$};
\node (v5) [label=left:$v_5$] at (2,0.5) {$\bullet$};
\node (v6) [label=right:$v_6$] at (6,0.5) {$\bullet$};

\draw[-latex] (v1) to node [auto,swap,inner sep=1pt] {$a$} (v2);
\draw[-latex] (v1) to node [auto,inner sep=1pt] {$b$} (v3);
\draw[-latex] (v1) to node [auto,inner sep=1pt] {$(ab)$} (v4);
\draw[-latex] (v2) to node [auto,swap,inner sep=1pt] {$b$} (v4);
\draw[-latex] (v3) to node [auto,inner sep=1pt] {$a$} (v4);
\draw[-latex] (v2) to node [auto,inner sep=1pt] {$a$} (v5);
\draw[-latex] (v5) to node [auto,inner sep=1pt] {$a$} (v6);
\draw[-latex] (v3) to node [auto,inner sep=1pt] {$b$} (v6);
\end{tikzpicture}
~~~~~~~~
\begin{tikzpicture}[scale=.8]
\node[circle,draw,minimum size=0.7cm] (p1) at (4,4) {\textit{6}};
\node[circle,draw,minimum size=0.7cm] (p2) at (4,0) { };
\node[circle,draw,minimum size=0.7cm] (p3) at (2,3.5) {$\bullet$};
\node[circle,draw,minimum size=0.7cm] (p4) at (6,3.5) {$\bullet$};
\node[draw,minimum size=0.6cm] (a) at (3,2){$a$};
\node[draw,minimum size=0.6cm] (b) at (5,2){$b$};
\draw[-latex] (p1) to node [auto,swap,inner sep=1pt] {\textit{2}} (a);
\draw[-latex] (p1) to node [auto,inner sep=1pt] {\textit{3}} (b);
\draw[-latex] (a) to node [auto,swap,inner sep=1pt] {\textit{2}} (p2);
\draw[-latex] (b) to node [auto,inner sep=1pt] {\textit{3}} (p2);
\draw[-latex] (p3) to node [auto,inner sep=1pt] { } (a);
\draw[-latex] (p4) to node [auto,inner sep=1pt] { } (b);
\draw[-latex] (a) to node [auto,inner sep=1pt] { } (p3);
\draw[-latex] (b) to node [auto,inner sep=1pt] { } (p4);
\end{tikzpicture}
~~
($b$)
\end{center}\vspace*{-2mm}
\caption{
    Splitting is not enough to guarantee
    reversing (Example~\ref{ex:splitneg}).
    Note that $v_1$ is the initial state.
\label{fig-ex2}}
\end{figure}
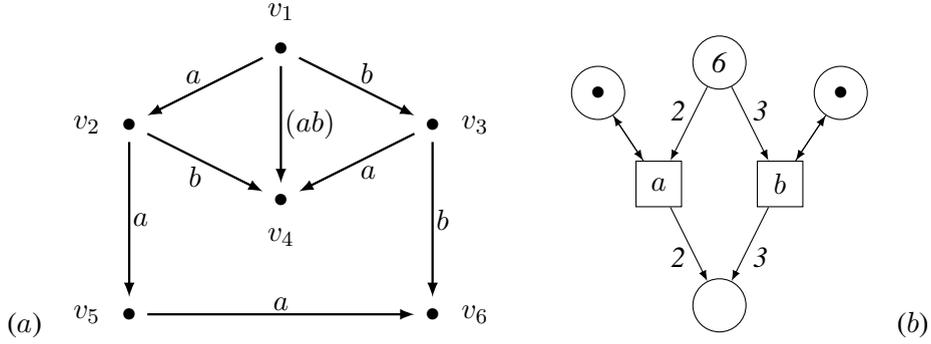

\medskip
    Let $p$ be any place of $N$.
    We first observe that the effect of
    executing the sequences of actions
    $aaa$ and $bb$ on $p$ is the same, when going from $M_1$ to $M_6$.
    Hence,
    $3\cdot\eff_N(a)(p)=2\cdot\eff_N(b)(p)$,
    and so there is an integer $k$ such that
    $\eff_N(a)(p)=2k$ and $\eff_N(b)(p)=3k$.
    With this observation, and by considering different arrows in $\STS$,
    we obtain:
\[
\begin{array}{l@{~}l@{~}l@{~~~~~~~~~~}l@{~}l@{~}l@{~~~~~~~~~~}l@{~}l@{~}l}
     M_2(p)
    & =
    & M_1(p)+2k
    & M_3(p)
    & =
    & M_1(p)+3k
    & M_4(p)
    & =
    & M_1(p)+5k
\\
     M_5(p)
    & =
    & M_1(p)+4k
    & M_6(p)
    & =
    & M_1(p)+6k\;.
\end{array}
\]
    Hence, in particular, we have:
\begin{equation}
\label{eq-ex2}
     M_3(p)\leq M_5(p)\leq M_4(p)
     ~~~~~~~~ \mbox{or} ~~~~~~~~
     M_3(p)\geq M_5(p)\geq M_4(p)\;.
\end{equation}
    Suppose now that $(\reverse{a}_{\IDX{i}}\reverse{b}_{\IDX{j}})$
    is a step reversing
    $(ab)$ at $M_4$.
    Then, by \textit{SEQ} and \textit{CE} holding
    for the concurrent reachability graphs of \textsc{pt}-nets,
    $\reverse{b}_{\IDX{j}}$ is also enabled at $M_3$.
    On the other hand, $\reverse{b}_{\IDX{j}}$ is not enabled at $M_5$.
    Then there must be a place $p$ of $N$ such that
    $M_5(p)<\PRE_N(\reverse{b}_{\IDX{j}})(p)$.
    But we also have $M_3(p)\geq \PRE_N(\reverse{b}_{\IDX{j}})(p)$ and
    $M_4(p)\geq \PRE_N(\reverse{b}_{\IDX{j}})(p)$,
    as $\reverse{b}_{\IDX{j}}$ is enabled at $M_3$ and $M_4$.
    This, however, produces a contradiction with Eq.(\ref{eq-ex2}).
\hfill$\diamondsuit$
\end{example}

Example~\ref{ex:splitneg} can be used further to show
that even
allowing inhibitor arcs in $N$ would not help.\footnote
{
    An inhibitor arc between a place $p$ and action $t$
    means that if $t$ is enabled at a marking $M$, then
    $M(p)=0$.
}
The reason is that due to the formulas
Eq.(\ref{eq-ex2}) for the markings $M_3$,
$M_4$, and $M_5$,
no inhibitor place $p$ could be empty at $M_3$ and $M_4$, and contain a
token at $M_5$.
It would therefore be useless to block $\reverse{b}_{\IDX{j}}$ at $M_5$ and still
allow the execution of $\reverse{b}_{\IDX{j}}$ at $M_3$ and $M_4$.
Thus, reversing using \textsc{pt}-nets with inhibitor arcs is also not
going to work in the general case, when considering the step semantics.
This justifies the need to use test arcs `stronger' than inhibitor arcs
in addition to the splitting of reverse actions.
Indeed, a general solution can then be obtained using
an extended model of \textsc{pt}-nets, as shown in the next section.

\section{A solution combining splitting and weighted read arcs}
\label{sect-split}

A \emph{\textsc{pt}-net with weighted read arcs}  (or \textsc{ptr}-net)
is a tuple $N=(P,T,F,R,M_0)$ such that $N'=(P,T,F,M_0)$
is a \textsc{pt}-net, and $R:P\times T\to\mathbb{N}$ is a partial function
defining \emph{read arcs}.
All the notations and concepts introduced for $N'$ are applicable to $N$ except that
a step $\alpha$ of $N$ is enabled at a marking $M$ if it is enabled at marking $M$
in $N'$ and, in addition, $R(p,t)=M(p)$,
whenever $a\in\alpha$ and $p\in P$ are such that $R(p,a)$ is defined.
Read arcs are depicted as arrows with square arrowheads and
labelled by their weights.

As the read arcs do not affect markings which result from firing
steps of actions, the concurrent reachability graphs of
\textsc{ptr}-nets satisfy \textit{CE}. Although \textit{SEQ}
may fail to hold, it is the case that if $\alpha$ is an enabled step,
then each step $\beta\leq\alpha$ is also enabled.

We first show that there is a~\textsc{pt}-net with weighted read arcs
reversing the reachability graph from Example~\ref{ex:splitneg}.

\begin{example}
\label{ex:wread}
    Recall the step transition system and the
    \textsc{pt}-net from Example~\ref{ex:splitneg}.
    The construction of a solution comes in two phases.
    In the first phase, splitting is used to reverse
    all singleton steps.
    The result, which uses two reverses for $a$ and two reverses for $b$,
    is shown in Figure~\ref{fig-exfw1}($b$).
    Note that although all the singleton steps are indeed reversed, the only
    non-singleton step $(ab$) is not.
    The second phase of
    the construction adds reverses for $a$ and $b$ which are
    simultaneously executable at $M_4$, as shown in Figure~\ref{fig-exfw1}($d$).
    A solution is then obtained by joining together
    Figures~\ref{fig-ex2}(b), \ref{fig-exfw1}($b$) and~\ref{fig-exfw1}($d$),
    by identifying the places with 6 tokens and the places with 0 tokens.
\hfill$\diamondsuit$
\end{example}

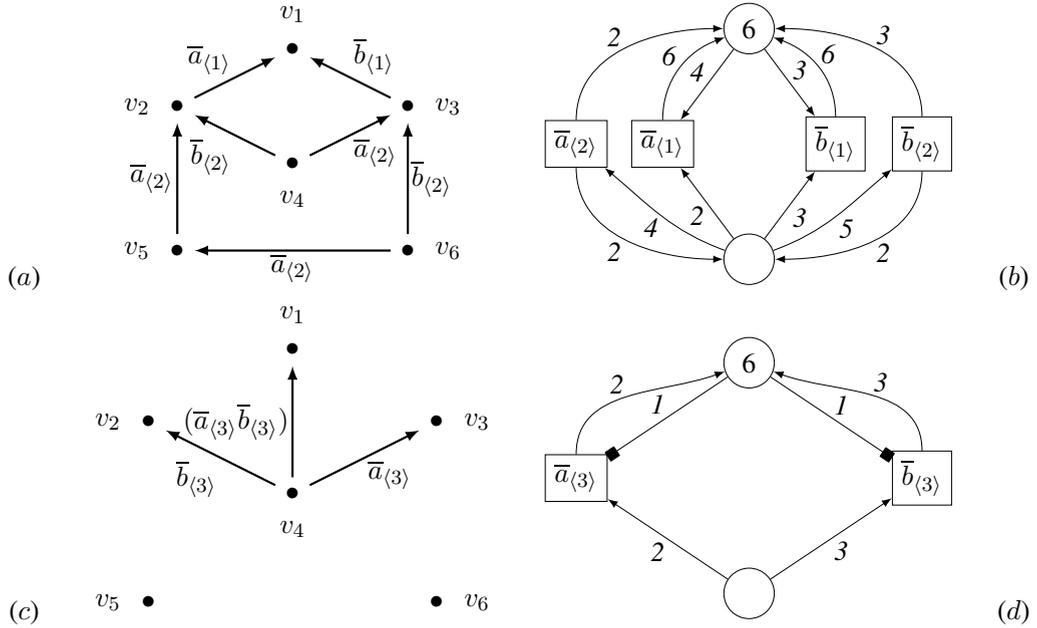
\begin{figure}[!h]
\vspace*{-1mm}
\begin{center}
\scalebox{0.95}{
\begin{tabular}{cccc}
($a$)
&
\begin{tikzpicture}[node distance=1.3cm,>=arrow30,%
     line width=0.3mm,scale=.8,bend angle=45]
     \tikzstyle{box}=[draw,regular polygon,thick,%
     regular polygon sides=4,minimum size=22mm, inner sep = -3pt]

\node (v1) [label=$v_1$]       at (4,4) {$\bullet$};
\node (v2) [label=left:$v_2$]  at (2,3) {$\bullet$};
\node (v3) [label=right:$v_3$] at (6,3) {$\bullet$};
\node (v4) [label=below:$v_4$] at (4,2) {$\bullet$};
\node (v5) [label=left:$v_5$]  at (2,0.5) {$\bullet$};
\node (v6) [label=right:$v_6$] at (6,0.5) {$\bullet$};

\draw[-latex] (v2) to node [auto,inner sep=1pt] {$\reverse{a}_{\IDX{1}}$} (v1);
\draw[-latex] (v3) to node [auto,swap,inner sep=1pt] {$\reverse{b}_{\IDX{1}}$} (v1);
\draw[-latex] (v4) to node [auto,inner sep=1pt] {$\reverse{b}_{\IDX{2}}$} (v2);
\draw[-latex] (v4) to node [auto,swap,inner sep=1pt] {$\reverse{a}_{\IDX{2}}$} (v3);
\draw[-latex] (v5) to node [auto,inner sep=1pt] {$\reverse{a}_{\IDX{2}}$} (v2);
\draw[-latex] (v6) to node [auto,inner sep=1pt] {$\reverse{a}_{\IDX{2}}$} (v5);
\draw[-latex] (v6) to node [auto,swap,inner sep=1pt] {$\reverse{b}_{\IDX{2}}$} (v3);
\end{tikzpicture}
&
\begin{tikzpicture}[scale=.8]
\node[circle,draw,minimum size=0.7cm] (p1) at (3,4) {6};
\node[circle,draw,minimum size=0.7cm] (p2) at (3,0) { };
\node[draw,minimum size=0.6cm] (a1) at (1.5,2){$\reverse{a}_{\IDX{1}}$};
\node[draw,minimum size=0.6cm] (a2) at (0,2){$\reverse{a}_{\IDX{2}}$};
\node[draw,minimum size=0.6cm] (b1) at (4.5,2){$\reverse{b}_{\IDX{1}}$};
\node[draw,minimum size=0.6cm] (b2) at (6,2){$\reverse{b}_{\IDX{2}}$};
\draw[-latex] (p1) to node [auto,swap,inner sep=1pt] {\textit{4}} (a1);
\draw[-latex] (a1) to [out=90, in=-160] node [auto,inner sep=1pt] {\textit{6}} (p1);
\draw[-latex] (a2) to [out=90, in=180] node [auto,inner sep=1pt] {\textit{2}} (p1);
\draw[-latex] (p2) to node [auto,inner sep=1pt] {\textit{2}} (a1);
\draw[-latex] (p2) to [out=160, in=-40]node [auto,inner sep=1pt] {\textit{4}} (a2);
\draw[-latex] (a2) to [out=-90, in=180]node [auto,swap,inner sep=1pt] {\textit{2}} (p2);
\draw[-latex] (p1) to node [auto,inner sep=1pt] {\textit{3}} (b1);
\draw[-latex] (b1) to [out=90, in=-20] node [auto,swap,inner sep=1pt] {\textit{6}} (p1);
\draw[-latex] (b2) to [out=90, in=0] node [auto,swap,inner sep=1pt] {\textit{3}} (p1);
\draw[-latex] (p2) to node [auto,swap,inner sep=1pt] {\textit{3}} (b1);
\draw[-latex] (p2) to [out=20, in=-140]node [auto,swap,inner sep=1pt] {\textit{5}} (b2);
\draw[-latex] (b2) to [out=-90, in=0]node [auto,inner sep=1pt] {\textit{2}} (p2);
\end{tikzpicture}
&
($b$)
\\
($c$)
&
\begin{tikzpicture}[node distance=1.3cm,>=arrow30,%
     line width=0.3mm,scale=1,bend angle=45]
     \tikzstyle{box}=[draw,regular polygon,thick,%
     regular polygon sides=4,minimum size=22mm, inner sep = -3pt]

\node (v1) [label=$v_1$]     at (4,4) {$\bullet$};
\node (v2) [label=left:$v_2$] at (2,3) {$\bullet$};
\node (v3) [label=right:$v_3$] at (6,3) {$\bullet$};
\node (v4) [label=below:$v_4$] at (4,2) {$\bullet$};
\node (v5) [label=left:$v_5$]  at (2,0.5) {$\bullet$};
\node (v6) [label=right:$v_6$] at (6,0.5) {$\bullet$};
\draw[-latex] (v4) to node [auto,inner sep=1pt] {$(\reverse{a}_{\IDX{3}}\reverse{b}_{\IDX{3}})$} (v1);
\draw[-latex] (v4) to node [auto,inner sep=1pt] {$\reverse{b}_{\IDX{3}}$} (v2);
\draw[-latex] (v4) to node [auto,swap,inner sep=1pt] {$\reverse{a}_{\IDX{3}}$} (v3);
\end{tikzpicture}
&
\begin{tikzpicture}[scale=.8]
\node[circle,draw,minimum size=0.7cm] (p1) at (4,4) {6};
\node[circle,draw,minimum size=0.7cm] (p2) at (4,0) { };
\node[draw,minimum size=0.6cm] (a1) at (1,2){$\reverse{a}_{\IDX{3}}$};
\node[draw,minimum size=0.6cm] (b1) at (7,2){$\reverse{b}_{\IDX{3}}$};
\draw[-square] (p1) to node [auto,swap,inner sep=1pt] {\textit{1}} (a1);
\draw[-latex] (a1) to [out=90, in=-160] node [auto,inner sep=1pt] {\textit{2}} (p1);
\draw[-latex] (p2) to node [auto,inner sep=1pt] {\textit{2}} (a1);
\draw[-square] (p1) to node [auto,inner sep=1pt] {\textit{1}} (b1);
\draw[-latex] (b1) to [out=90, in=-20] node [auto,swap,inner sep=1pt] {\textit{3}} (p1);
\draw[-latex] (p2) to node [auto,swap,inner sep=1pt] {\textit{3}} (b1);
\end{tikzpicture}
&
($d$)
\end{tabular} }
\end{center}\vspace*{-5mm}
\caption{
    Reversing with splitting: phase one ($a,b$), and  phase two ($c,d$).
\label{fig-exfw1}}
\end{figure}

The solution presented in Example~\ref{ex:wread} inspired
the development of a
general construction which works
for an arbitrary bounded \textsc{pt}-net.

Let $N=(P,T,F,M_0)$ be a bounded \textsc{pt}-net,
and let $n$ be an upper limit on
the sizes of steps enabled at its reachable
markings (such an $n$ always exists as the concurrent reachability
graph of $N$ is finite).
Moreover, for every marking $M\in\reach_N$,
the steps
annotating actions incoming to $M$ in the concurrent reachability
graph are
$\income_N(M)=\{\alpha\mid \exists M'\in\reach_N:
M'\xright{\alpha}_NM\}$.
Since $\CRG_N$ is a \textsc{cest}-system,
$\alpha\leq\beta\in\income_N(M)$ implies
$\alpha\in\income_N(M)$.

We then construct a \textsc{ptr}-net
$N'=(P\uplus P',T\uplus T',F\munion F', R,M_0\munion M'_0)$.
A key aspect of the construction is that for each
reachable marking $M$ of $N$, and for each maximal step\footnote
{
    That is, $\alpha\leq\beta\in\income_N(M)$
    implies $\alpha=\beta$.
}
$\alpha\in\income_N(M)$,
we add a set of fresh
actions
$T_{\alpha,M}
 =
 \{\reverse{a}_{\IDX{\alpha,M,i}}
 \mid
 a\in\alpha\wedge 1\leq i\leq\alpha(a)\}$.
We then proceed thus:
\begin{itemize}
\item
    For every new action $\reverse{a}_{\IDX{\alpha,M,i}}\in T'$:
\begin{itemize}
\item
    $\PRE_{N'}(\reverse{a}_{\IDX{\alpha,M,i}})|_P=\POST_N(a)$ and
    $\POST_{N'}(\reverse{a}_{\IDX{\alpha,M,i}})|_P=\PRE_N(a)$.

\item
    For every $b\in T$, we add a fresh (mutex) place,
    as in Figure~\ref{fig-ex3}($a$).

\item
    For every $\reverse{b}_{\IDX{\beta,M,j}}\in T'$
    with $\alpha\neq\beta$,
    we add a fresh (mutex) place, as in Figure~\ref{fig-ex3}($b$).
\end{itemize}

\item
    $P\times T'$ is the domain of $R$ and
    $R(p,\reverse{a}_{\IDX{\alpha,M,i}})=M(p)$, for all
    $p\in P$ and $\reverse{a}_{\IDX{\alpha,M,i}}\in T'$.

\item
    $M'_0\in\mult(P')$ is the marking of
    the places in $P'$ as indicated in Figure~\ref{fig-ex3}.
\end{itemize}

\begin{figure}[!ht]
\begin{center}
($a$)
~~
\begin{tikzpicture}[scale=1]
\node[circle,draw,minimum size=0.7cm] (p) at (2,0) {$n$};
\node[draw,minimum size=0.7cm] (a) at (0,0){$b$};
\node[draw,minimum size=0.7cm] (b) at (4,0){$\reverse{a}_{\IDX{\alpha,M,i}}$};

\draw[-latex] (p) to [out=30, in=150] node [auto,inner sep=2pt] {$n$} (b);
\draw[-latex] (b) to [out=-150, in=-30] node [auto,inner sep=2pt] {$n$} (p);

\draw[-latex] (a) to node [auto,inner sep=1pt] { } (p);
\draw[-latex] (p) to node [auto,inner sep=1pt] { } (a);

\node[circle,draw,minimum size=0.7cm] (p1) at (9,0) {$\bullet$};
\node[draw,minimum size=0.7cm] (a1) at (7,0){$\reverse{a}_{\IDX{\alpha,M,i}}$};
\node[draw,minimum size=0.7cm] (b1) at (11,0){$\reverse{b}_{\IDX{\beta,M,j}}$};

\draw[-latex] (p1) to node [auto,inner sep=1pt] { } (a1);
\draw[-latex] (p1) to node [auto,inner sep=1pt] { } (b1);
\draw[-latex] (a1) to node [auto,inner sep=1pt] { } (p1);
\draw[-latex] (b1) to node [auto,inner sep=1pt] { } (p1);
\end{tikzpicture}
~~
($b$)
\end{center}\vspace*{-2mm}
\caption{Places $P'$ added in the construction of $N'$.
\label{fig-ex3}}
\end{figure}
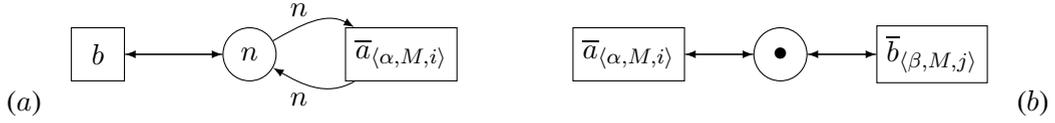

We then obtain the desired result.

\begin{theorem}
    $\CRG_{N'}$ is a split reverse of $\CRG_N$.
\end{theorem}
\begin{proof}
    Let $\STS=\CRG_N$ and $\STS'=\CRG_{N'}$.
    We first gather together some immediate facts
    about  $N'$.

\begin{lemma}
\label{lem-rmc}
    {\ }
\begin{enumerate}
\item
    $\reverse{a}_{\IDX{\alpha,M,i}}$ is an
    indexed   reverse of  $a$, for all
    $\reverse{a}_{\IDX{\alpha,M,i}}\in T'$ and $a\in T$.

\item
    $\eff_{N'}(\alpha) =\eff_N(\alpha)\munion\es_{P'}$,
    for every $\alpha\in\mult(T)$.

\item
    $\eff_{N'}(\gamma)= -\eff_N(\alpha)\munion\es_{P'}$,
    for all $\gamma\in\mult(T')$ and $\alpha\in\mult(T)$ such that
    $\ralpha=\noidx(\gamma)$.

\item
    $M|_{P'}=M'_0$, for every $M\in\reach_{N'}$.

\item
    If $\gamma$ is a step enabled at
    $M\in\reach_{N'}$, then $\gamma\in\mult(T)$,
    or there is $\alpha\in\income_N(M)$
    such that $\gamma$ is a set included in $T_{\alpha,M}\subseteq T'$.
\end{enumerate}
\end{lemma}
\begin{proof}[Lemma~\ref{lem-rmc}]
(1,2)
    Follow directly from the definition of
    $N'$.

(3) Follows from part (1).

(4) Follows from parts (2) and (3).

(5) By part (4), $M|_{P'}=M'_0$.
    Hence the result follows from the presence of the
    weighted read arcs $R$ and
    the mutex places shown in Figure~\ref{fig-ex3}.
\end{proof}

    We will show that $\reach_{N'}=\{M\munion M_0'\mid M\in\reach_N\}$
    and
    $\STS^\rev\simeq_\psi\noidx(\STS')$,
    where $\psi(M)=M\munion M_0'$, for every $M\in\reach_N$.

    We first observe that
    $\psi(M_0)= M_0\munion M_0'$ is the initial marking of $N'$.
    Suppose  that
    $M\in\reach_N$ is such that $\psi(M)=M\munion M_0'\in\reach_{N'}$.
    To show that the executions of steps are
    preserved by $\psi$ in both directions,
    we consider four cases, after taking into account
    Lemma~\ref{lem-rmc}(5).

\medskip
\textit{Case 1:}
    $M\xright{\alpha}_\STS M'$.
    Then,
    since $n$ in Figure~\ref{fig-ex3}($a$)
    is such that $|\alpha|\leq n$,
    the addition of the new places $P'$ does not block $\alpha$.
    Hence $\alpha$ is
    enabled at $M\munion M_0'$.
    Moreover, by Lemma~\ref{lem-rmc}(2),
    $M\munion M_0'\xright{\alpha}_{\STS'} M'\munion M_0'$.

\textit{Case 2:}
    $M\xright{\ralpha}_{\STS^\rev} M'$.
    Then $M'\xright{\alpha}_\STS M$ and $\alpha\in\income_N(M)$.
    Let $\beta$  be any maximal step in $\income_N(M)$ such that
    $\alpha\leq\beta$ (such a step exists since $\CRG_N$ is finite).
    Then there is a subset $\gamma$ of $T_{\beta,M}$ such that
    $\noidx(\gamma)=\ralpha$.
    By construction, $\gamma$ is enabled at $M\munion M_0'$.
    Hence, by Lemma~\ref{lem-rmc}(3),
    $M\munion M_0'\xright{\gamma}_{\STS'} M'\munion M_0'$.

\textit{Case 3:} $M\munion M_0'\xright{\alpha}_{\STS'} M'$ and
    $\alpha\in\mult(T)$.
    Then, by construction and Lemma~\ref{lem-rmc}(2), $\alpha$ is
    enabled at $M$
    and $M'=(M+\eff_N(\alpha))\munion M_0'$.
    Moreover,
    $M \xright{\alpha}_{\STS^\rev} M+\eff_N(\alpha)$.

\textit{Case 4:} $M\munion M_0'\xright{\gamma}_{\STS'} M'$,
    where $\gamma$ is a subset of $T_{\alpha,M}$ for some
    $\alpha\in\income_N(M)$.
    Let $\beta=\noidx(\gamma)\leq\ralpha$. Then, by construction
    and Lemma~\ref{lem-rmc}(3),
    $M'=(M-\eff_N(\beta))\munion M_0'$, $\beta$ is enabled
    at $M-\eff_N(\beta)$, and $M-\eff_N(\beta)\xright{\beta}_\STS M$.
    Hence
    $M \xright{\reverse{\beta}}_{\STS^\rev} M-\eff_N(\beta)$.
\end{proof}

We have developed a general construction which brings us to the same
level of reversibility as in the sequential case.
However, we had to pay the (costly)
price of using of a non-standard class of
read arcs.
The construction presented above
is far from being optimal. Taking as an example the solution
from Example~\ref{ex:wread}, we observe that it
would introduce 5 reverses of $a$, 4 reverses of $b$,
and a total of 31 additional places.
One can easily see that a large number of
them could be avoided, by considering
the conditions that force the introduction
of each split reversal and those requiring the
addition of the new control places.
We expect that the proposed construction could be optimised by reducing the number of split reverses and, at the same time, allowing some them to exhibit autoconcurrency (which is admitted in \textsc{ptr}-nets).

\section{Concluding remarks}
\label{sect-tdt}

In this paper, we continued a
study of reversibility in
\textsc{pt}-nets, when the step semantics
based on executing steps (multisets) of
actions rather than single actions is considered,
thus capturing \emph{real parallelism}.
In a more abstract setting, the (partial)
reversal of steps, thus generating \emph{mixed
steps} possibly containing both original and
reverse action, has been studied
in~\cite{PhiUl15}. Here we discussed
how such reversing can be done in a concrete
operational framework of \textsc{pt}-nets.

In the future work, we plan to develop an effective
solution to the synthesis problem for the step transition
systems with multiple initial states, and
address the optimisation of the general solution based on
\textsc{ptr}-nets presented in the last section.

\paragraph{Acknowledgement}

We would like to thank the anonymous reviewers for their careful reading many insightful comments and suggestions for improvement.

This research was supported by the EU Cost Action IC1405 and
the Polish National Agency for Academic Exchange
under Grant PPI/APM/2018/1/00036/U/001.
The first author was partially supported by projects TIN2015-67522-C3-3-R,
PID2019-108528RB-C22, and by Comunidad de Madrid as part of
the program S2018/TCS-4339 (BLOQUES-CM) co-funded by EIE Funds of the European Union.
The third author was partially supported by the
Polish NCN Grant 2017/27/B/ST6/02093.

\end{document}